\documentclass[a4paper, USenglish, autoref, cleveref, thm-restate]{lipics-v2021}
\nolinenumbers
\hideLIPIcs\nolinenumbers\pdfoutput=1 

\usepackage{tikz}
\usetikzlibrary{positioning, cd}
\usepackage{booktabs}
\usepackage{mathtools}

\newcommand{\defn}[1]{\textbf{\textup{#1}}}

\newtheorem*{convention*}{Convention}

\usepackage{complexity} \newclass{\pFO}{pFO}
\newclass{\Comp}{Comp}
\newclass{\RTIME}{DTIME^R}
\newcommand{\para}[1]{\textup{para-}#1}
\renewcommand{\p}[1]{\ComplexityFont{p}#1}
\newcommand{\up}{\uparrow}
\DeclareMathOperator{\Para}{Para}

\newcommand{\gate}[1]{\textup{\textsc{#1}}}
\def\gand/{\gate{and}}
\def\gor/{\gate{or}}
\def\gnot/{\gate{not}}

\DeclareMathOperator{\bit}{Bit}
\newcommand{\N}{\mathbb N}
\renewcommand{\implies}{\mathbin{\rightarrow}}
\renewcommand{\iff}{\mathbin{\leftrightarrow}}
\DeclareMathOperator{\double}{\delta} \DeclareMathOperator{\fanin}{In} 

\title{Uniformity within Parameterized Circuit Classes}
\subtitle{An Exercise for the Reader}

\author{Steef Hegeman}{Leiden University, Netherlands}{s.hegeman@liacs.leidenuniv.nl}{0009−0009−0368−8502}{}
\author{Jan Martens}{Leiden University, Netherlands}{j.j.m.martens@liacs.leidenuniv.nl}{0000-0003-4797-7735}{}
\author{Alfons Laarman}{Leiden University, Netherlands}{a.w.laarman@liacs.leidenuniv.nl}{0000-0002-2433-4174}{}
\authorrunning{S. Hegeman, J.\,J.\,M. Martens, and A.\,W. Laarman}
\Copyright{Steef Hegeman, Jan Martens, and Alfons Laarman}

\ccsdesc{Theory of computation~Complexity classes}
\ccsdesc{Theory of computation~Circuit complexity}
\ccsdesc{Theory of computation~Fixed parameter tractability}
\keywords{Parameterized complexity, circuit complexity, uniformity, descriptive complexity}

\begin{document}

\maketitle

\begin{abstract}
We study uniformity conditions for parameterized Boolean circuit families.
Uniformity conditions require that the infinitely many circuits in a circuit family are in some sense easy to construct from one shared description.
For shallow circuit families,
logtime-uniformity is often desired but quite technical to prove.
Despite that, proving it is often left as an exercise for the reader
-- even for recently introduced classes in parameterized circuit complexity,
where uniformity conditions have not yet been explicitly studied.
We formally define parameterized versions of linear-uniformity, logtime-uniformity, and FO-uniformity, and
prove that these result in equivalent complexity classes
when imposed on $\para\AC^0$ and $\para\AC^{0\uparrow}$.
Overall, we provide a convenient way to verify uniformity for shallow
parameterized circuit classes, and thereby substantiate claims of uniformity in
the literature.
\end{abstract}

\newpage 

\section{Introduction}\label{sec:introduction} 
Circuit complexity theory analyzes problems through the
(smallest) size and depth of Boolean circuits that solve them,
linking logic, space complexity, and parallel complexity theory \cite{vollmer_introduction_1999}.
Because a Boolean circuit has a fixed number of inputs, the actual object of
study is a circuit family, which contains a circuit for every input length $n$.

Without restrictions, circuit families are too expressive to be used as a model of computation, as non-uniform
circuit families can even express non-computable functions. Hence,
uniformity conditions are often imposed to ensure that the circuits in the
family are in some sense easy to construct.
With uniformity conditions,
circuit families are a useful model in the study of,
amongst others,
parallel
and descriptive complexity theory \cite{vollmer_introduction_1999}.

The main idea is that, for a class of circuit families $\mathcal C$,
the structure of a family's circuits
should be no more complex than the problems that families in $\mathcal C$ can solve
-- preferably, strictly less complex.
For the class of constant-depth circuit families,
this line of thinking has led to the very restrictive logtime-uniformity condition,
requiring that the connection language of the circuits in a family must be (simultaneously) decidable in logarithmic time.

Logtime computations are difficult to analyze or work with~\cite{barrington_uniformity_1990,vollmer_introduction_1999}.
Therefore, logtime-uniformity is rarely proved directly,
but instead shown using results such as those of
Barrington, Immerman, and Straubing \cite{immerman_descriptive_1999}.
They show that if a constant-depth circuit family is describable in terms of first-order logic,
there is another, logtime-uniform constant-depth circuit family equivalent with it.
Using this correspondence, one can show the existence of a logtime-uniform circuit family
without directly dealing with the technical intricacies of logtime computations.

Recently, the interest in circuit complexity was renewed in the setting of
parameterized
complexity~\cite{bannach_fast_2015,bannach_parallel_2017,chen_slicewise_2017,pilipczuk2018parameterized,chen_lower_2019}.
In parameterized complexity theory, problems are classified on a finer scale with respect to a
parameter instead of just their size. 
In the context of circuit complexity this spawns new circuit classes, that
identify problems which are in some sense parallel constant time decidable up-to some
parameter~\cite{bannach_parallel_2017}.  

Parameterized circuit classes require new parameterized uniformity conditions.
Recent works \cite{bannach_fast_2015,bannach_parallel_2017,chen_slicewise_2017,pilipczuk2018parameterized,chen_lower_2019}
have adapted the logtime-uniformity of Barrington et al.,
though they are not very specific in their exact definitions.
They generally refer to the complexity of computing single bits in the binary encoding of circuits,
but it is not specified what the binary encoding is.
Additionally, though the literature cites \cite{barrington_uniformity_1990},
it is not clear if the equivalence results of Barrington et al.
carry over to the parameterized circuit classes.

In general, the literature often leaves uniformity claims as an exercise for the reader.
To the best of our knowledge, only \cite[Theorem 3.2, Appendix A]{bannach_fast_2015}
proves one of its uniformity claims.
However, the proof appears to erroneously assume that logtime-uniformity
gives a random-access Turing machine polylogarithmic time to
decide the graph structure of the circuits,
while it allows only logarithmic time.

\begin{table}[hb]
	\caption{
		Informal overview of the parameterized uniformity notions used in the paper.
	}\label{tab:uniformity-notions-rough}
	\centering
	\begin{tabular}{@{}llll@{}} \toprule
		Condition & Question (for circuit $C_{n,k}$) & Decision complexity \\ \midrule
		linear-BD-uniform
			& $\langle \text{local question}, n, k \rangle$
			& $\DTIME(\log n + f(k))$ \\
		logtime-D-uniform
			& $\langle \text{local question}, \{0,1 \}^n, \{0,1 \}^k \rangle$
			& $\RTIME(\log n + f(k))$ \\
		logtime-E-uniform
			& $\langle \text{path question}, \{0,1 \}^n, \{0,1 \}^k \rangle$
			& $\RTIME(\log n + f(k))$ \\
		$\FO$-D-uniform
			& $\langle \text{local question}, \{0,1 \}^n, \{0,1 \}^k \rangle$
			& $\para\FO$ in parameter $k$ \\
		$\FO$-E-uniform
			& $\langle \text{path question}, \{0,1 \}^n, \{0,1 \}^k \rangle$
			& $\para\FO$ in parameter $k$ \\
		\bottomrule
	\end{tabular}
\end{table}

\subparagraph*{Contributions.}
We introduce and analyze uniformity conditions for parameterized circuit
complexity. 
Specifically, we define a suitable notion of the logtime-uniformity used in the
literature, and show it to be equivalent to more natural notions when imposed on
the shallow parameterized circuit classes $\para\AC^0$ and $\para\AC^{0\up}$.

An overview of the discussed uniformity notions,
formally defined in \cref{sec:direct-uniformity,sec:extended-uniformity},
can be found in \autoref{tab:uniformity-notions-rough}.
We consider direct uniformity notions, which require the local structure of
circuits to be easily decidable, and extended uniformity notions that also
consider paths in the circuit.
The direct logtime-D-uniformity is based on logtime-DCL-uniformity introduced in \cite{barrington_uniformity_1990}
and used in the reference work \cite{vollmer_introduction_1999}.
We believe our definition captures the more informal notions of parameterized logtime-uniformity found in \cite{bannach_fast_2015,bannach_parallel_2017,chen_slicewise_2017,pilipczuk2018parameterized,chen_lower_2019}.
For the classes $\para\AC^0$ and $\para\AC^{0\up}$,
we show that these notions are equivalent to linear-BD-uniformity,
a parameterized variant of $U_D$-uniformity from \cite{ruzzo_uniform_1981}.
We introduce a parameterized variant of $\FO$-uniformity from~\cite{barrington_uniformity_1990},
and show that
this too is equivalent to our earlier notions when applied to $\para\AC^0$ and $\para\AC^{0\up}$.
Finally, we also show this
for strict adaptations of extended logtime and first-order uniformity.

Our results allow for an easier analysis of uniformity for shallow parameterized circuit classes,
and address uniformity claims in the literature.

\subparagraph*{Outline.}
\cref{sec:preliminaries} contains preliminaries on circuit complexity and known uniformity results.
In \cref{sec:direct-uniformity},
we formally define the direct uniformity conditions roughly described in \autoref{tab:uniformity-notions-rough}.
In \cref{sec:para-ac0-and-para-ac0^},
we show that imposing these different conditions leads to the same complexity class,
both for $\para\AC^0$ and $\para\AC^{0\up}$.
In \cref{sec:extended-uniformity},
we show that this also holds for strict adaptations of extended uniformity conditions.
In \cref{sec:descriptive},
we show how our work can be related to the descriptive complexity-based work on uniformity for $\AC[t(n)]$.

\section{Preliminaries}\label{sec:preliminaries}
In this section,
unless otherwise noted,
standard definitions in circuit complexity
and parameterized complexity are based on the reference works
\cite{vollmer_introduction_1999} and \cite{flum_parameterized_2006} respectively.

\subsection{Circuit families}
A \defn{circuit family} is  a set of (Boolean) circuits $\{ C_n : n \in \N \}$
such that $C_n$ has $n$ input gates (fan-in 0) and at least one output gate (fan-in 1, fan-out 0).
The remaining (\defn{internal}) gates are \gand/, \gor/, or \gnot/ gates.
A circuit's \defn{size} is the number of gates in it,
and its \defn{depth} the length of the longest (input to output) path it contains.
The size and depth of a circuit family $C$ are functions $s, d$,
where $s(n),d(n)$ are equal to the size and depth of $C_n$ respectively.

We generally care about circuit families in which each circuit has one output
gate. We then define the language recognized by the circuit family $C$ as
$\{ x \in 2^* : C_{|x|}(x) = 1 \}$. That is, given some binary string $x$, we
take the circuit $C_{|x|}$ with $|x|$ input gates, set the input gates to the
bits of $x$, and evaluate it. The word is accepted if and only if the output
gate evaluates to 1. More generally, a circuit with $n$ inputs and $m$ outputs
induces a binary function $2^n \to 2^m$. Circuit families induce a function $2^* \to
2^*$.

The \defn{circuit class} $\AC^0$ consists of all languages that are recognized
by polynomially sized circuit families of constant depth.
That is, $Q$ is in $\AC^0$ if there is a circuit family $C$
with size $s(n) = O(n^c)$ and depth $d(n) = O(1)$,
such that $Q = \{ x : C(x) = 1 \}$.
It is important to note that all gates have unbounded fan-out,
and the \gand/ and \gor/ gates unbounded fan-in.

\subsection{Parameterized complexity}
Parameterized complexity theory studies computational problems modulo (pre)computations on some input-based parameter.
The parameter function is considered part of the problem.

\begin{definition}
	A \defn{parameterized problem} $(Q, \kappa)$
	consists of a language $Q \subseteq 2^*$ and a
	polynomial-time computable
	\defn{parameter function} $\kappa: 2^* \to \N$.
\end{definition}

We follow the literature on shallow parameterized circuit complexity \cite{elberfeld_space_2015,bannach_fast_2015}
in additionally requiring $\kappa$ to be $\FO$-computable \cite[Definition 2.5]{immerman_descriptive_1999},
equivalently that $\kappa$ can be computed by a logtime-D-uniform circuit family of polynomial size and constant depth
(cf. \cref{def:logtime-D-uniform}).
This requirement can be removed by modifying the definition of $\para\FO$ to being given $f(\kappa(x))$ as a precomputed value
instead of having to recognize it,
as discussed in \autoref{sec:kappa-fo}.

Parameterized classes are often defined in terms of non-parameterized classes.

\begin{definition}[\cite{flum_describing_2003}]\label{def:Para}
	Let $\mathcal X$ be a complexity class.
	A parameterized problem $(Q,\kappa)$ is in $\Para(\mathcal X)$
	if and only if
	there is a computable function $f$ with
	$\{ \langle x, f(\kappa(x)) \rangle : x \in Q \}$ in $\mathcal X$.
\end{definition}

Classes such as parameterized linear time can also be given more explicitly.

\begin{definition}\label{def:para-dtime}
	We write $\DTIME(t(n) + f(k))$ for the class of parameterized problems $(Q, \kappa)$
	for which there is a multitape Turing machine that solves
	$x \in Q$ in $O(t(|x|) + f(\kappa(x)))$ time.
\end{definition}

The difference in style between these definitions lies in whether the precomputation is part of the input or given implicitly by means of extra time.
In \cite{flum_describing_2003},
it is shown that $(Q,\kappa)$ is in $\DTIME(t(n) + f(k))$ for some computable $f$
if and only if $(Q,\kappa)$ is in $\para\DTIME(t(n))$,
but there are no such results for the sublinear random-access Turing machines we use later.

\subsection{Parameterized circuit complexity}
Parameterized circuit complexity studies classes $\Para(\mathcal X)$
for ordinary circuit complexity classes $\mathcal X$,
as well as other classes defined in terms of parameterized circuit families.

\begin{definition}[\cite{bannach_fast_2015}]
	A \defn{parameterized circuit family} $C$
	is a set $\{ C_{n,k} : n,k \in \N \}$ of circuits,
	such that each $C_{n,k}$ has $n$ input gates.
	A parameterized problem $(Q,\kappa)$ is decided by $C$
	if $x \in Q \Longleftrightarrow C_{|x|,\kappa(x)}(x) = 1$ holds.
	Size and depth are now functions $s(n,k)$ and $d(n,k)$.
\end{definition}

The parameterized circuit class $\para\AC^0$
is defined as the parameterized problems solved by constant-depth parameterized circuit families
of size $O(f(k)n^c)$, where $f$ is again some computable function that varies per circuit family.
The class $\para\AC^{0\up}$ is defined in terms of $O(f(k)n^c)$-sized circuits
with depth $O(f(k))$ \cite{bannach_fast_2015}.
$\Para(\AC^0)$ is equal to $\para\AC^0$ \cite{elberfeld_space_2015}.

\subsection{Uniformity}\label{subsec:uniformity}
Uniformity conditions for shallow circuit classes are defined in terms of
connection languages.
These languages express the connections between the gates in the circuit family.
To achieve this goal,
the gates in each circuit are first identified by a numbering.

\begin{definition}\label{def:numbering}
	An \defn{admissible numbering} of a circuit family
	$C = \{ C_n : n \in \N \}$
	of size $s(n)$
	assigns a number to every gate in $C$ so that:
	\begin{enumerate}
		\item within each $C_n$, no two gates have the same number,
		\item the $n$ input gates of $C_n$ are numbered from 0 to $n-1$,
		\item the $m$ output gates of $C_n$ are numbered from $n$ to $n+m-1$, and,
		\item gates in $C_n$ are numbered polynomially in $s(n)$,
			that is, with numbers of $O(\log s(n))$ bits.
	\end{enumerate}
\end{definition}

To improve readability, we will assume that circuit families come with fixed
admissible numberings (which we may choose freely),
and often identify gates with their gate number.

\begin{definition}[\cite{barrington_uniformity_1990,ruzzo_uniform_1981}]\label{def:direct-connection-languages}
	Let $C$ be a circuit family.
	The \defn{direct connection language} $L_D(C)$
	consists of all strings of the form
	$\langle G, a, p, z \rangle$
	where, writing $n$ for $|z|$:
	\begin{enumerate}
		\item $G$ is a gate in $C_n$;
		\item if $p = \epsilon$
			then $a < 3$ indicates the type of $G$
			(one of \gand/, \gor/, \gnot/);
		\item otherwise, the $p$th gate feeding into $G$ is $a$.
	\end{enumerate}

	The \defn{binary direct connection language} $L_{BD}(C)$
	consists of all strings of the form
	$\langle G, a, p, n \rangle$
	satisfying the same conditions,
	now with $n$ written in binary.	That is:
	\[
		L_{BD}(C) = \{\langle G,a,p,\lvert z\rvert \rangle : \langle G, a,p,z \rangle \in L_{D}(C)\}.
	\]
\end{definition}

There is also an extended connection language $L_E(C)$,
wherein $p$ can also be a path.
We will come back to this in \autoref{sec:extended-uniformity}.
The encoding of the lists $\langle \dots,  \dots \rangle$ in \cref{def:direct-connection-languages}
is addressed in \cref{subsec:lists}.
First, we introduce three direct uniformity notions,
defined in terms of
the decision complexity of direct connection languages.

\subparagraph*{Linear-time.}
In \cite{ruzzo_uniform_1981},
Ruzzo introduced
uniformity conditions for shallow circuit classes.
The following definition is based on Ruzzo's $U_D$-uniformity.\footnote{
	Ruzzo requires $\DTIME(\log s(n))$ on a slightly different language,
	but this is equivalent for our cases.
}

\begin{definition}
	A circuit family $C$ is
	\defn{linear-BD-uniform}
	if $L_{BD}(C)$ is in $\DTIME(n)$.
\end{definition}

That is,
a deterministic multitape Turing machine can decide $w \in L_{BD}(C)$ in time $O(|w|)$.

\subparagraph*{Logtime.}
Later uniformity notions were based on random-access Turing machines.
A \defn{random-access Turing machine} \cite[Appendix~A4]{vollmer_introduction_1999}
is a multitape Turing machine
equipped with a special query tape and query state.
Whenever the machine enters the query state,
it obtains the value of the $a$th bit on the input tape,
where $a$ is the binary address written on the query tape,
or an out-of-bounds response if $a$ is beyond the length of the input (or not a number).
Importantly,
the query tape is not erased on query,
and the input tape is read-only.

We write $\RTIME(t(n))$ for the class of problems that can be solved in $O(t(n))$ time by a random-access Turing machine,
and define $\RTIME(t(n) + f(k))$ analogously to \autoref{def:para-dtime}.

\begin{definition}[\cite{barrington_uniformity_1990}]\label{def:logtime-D-uniform}
		Circuit family $C$ is \defn{logtime-D-uniform} if $L_D(C)$ is in $\RTIME(\log n)$.
\end{definition}

Logtime-D-uniformity is arguably the most well-known uniformity condition for $\AC^0$.
Yet not much is known about $\RTIME(\log n)$.
Clearly $\RTIME(\log n) \subseteq \DTIME(n\log n)$ \cite{regan_gap-languages_1997},
but it is to the best of our knowledge unknown whether this is strict.
As far as we know,
it is for example
also open whether a random-access Turing machine can multiply two $m$-bit numbers in $O(m)$ time.
What is known is that they can add $m$-bit numbers in $O(m)$ time,
and determine the length of their input in $O(\log n)$ time \cite[Lemma 6.1]{barrington_uniformity_1990}.

\subparagraph*{First-order definable.}
The class $\FO$ consists of the languages that are first-order definable in the predicates $\leq$ and $\bit$.
That is, they are defined by a first-order sentence using symbols from
$\{ \forall, \land, \lnot, \leq, \bit, X \}$
as follows \cite{barrington_uniformity_1990,immerman_descriptive_1999}.
A word $w$ induces a first-order structure $S_w$ whose domain consists of the numbers 0 to $|w|$,
where:
$\leq$ is interpreted as the standard order on natural numbers,
$\bit(i,j)$ is true if and only if the $i$th bit of $j$ is 1,
and $X(i)$ is true if and only if the $i$th bit of $w$ is 1.
A sentence $\phi$ then defines the language $\{ w : S_w \models \phi \}$.

We may assume that the language also contains ternary predicates for addition and multiplication,
as they are first-order definable in $\leq$ and $\bit$ \cite{barrington_uniformity_1990,immerman_descriptive_1999}.
We may also assume that it contains a constant symbol for the length of the input,
and that we may quantify over (the binary representations of) numbers below $O(|w|^c)$,
by using $O(c)$ variables \cite[p. 293]{barrington_uniformity_1990}.

\begin{definition}[\cite{barrington_uniformity_1990}]
	A circuit family $C$ is \defn{\FO-D-uniform} if $L_D(C)$ is in $\FO$.
\end{definition}

For polynomially-sized circuit families,
it can be shown that linear-BD-uniformity implies logtime-D-uniformity,
which in turn implies \FO-D-uniformity.
\FO-D-uniformity is easier to work with than logtime-D-uniformity.
In fact, there are many examples of languages for which
the natural circuit families deciding them are quickly seen to be \FO-D-uniform,
but for which it is unknown whether they are logtime-D-uniform or linear-BD-uniform.

\begin{example}\label{ex:uniformity}
	Consider problem $X = \{ x : \text{the $\lfloor \sqrt{|x|} \rfloor$th bit of $x$ is 1} \}$.
	It is easy to see that $X$ is in $\AC^0$:
	there is a circuit family $C$ where each $C_n$
	consists of essentially one wire from the
	$\lfloor \sqrt{n} \rfloor$th input gate to the output gate.
	It has only one admissible numbering,
	and its direct connection language is $\{ \langle n, \lfloor \sqrt n \rfloor, 0, z \rangle : |z| = n \}$.

	The most difficult part in showing that this family is uniform
	is to show that we can somehow identify $\lfloor \sqrt{n} \rfloor$.
	For \FO-D-uniformity, this is easy:
	$r = \lfloor \sqrt n \rfloor$ is defined by the formula $r \cdot r \leq n \land (\forall s > r)\  s \cdot s > n$.
	Showing that $C$ is logtime-D-uniform appears to be more difficult,
	and we do not know whether this is the case.
	It is likely not linear-BD-uniform,
	as multiplication is conjectured to not be linear-time computable \cite{harvey_integer_2021}.
	\lipicsEnd
\end{example}

\subsection{Uniform \texorpdfstring{$\AC^0$}{AC0}}\label{sec:uniform-ac0}
Rather than the uniformity of specific circuit families,
we are usually more interested in whether an equivalent uniform family exists
for the same complexity class.

\begin{convention*}\label{conv:uniform-classes}
	Let $\mathcal U$ be a uniformity condition,
	and $\mathcal C$ a class of (parameterized) circuit families.
	If $\mathcal X$ is defined as the class of problems solved by families in $\mathcal C$,
	then \defn{$\mathcal U$-uniform~$\mathcal X$} is the subclass of problems
	solved by $\mathcal U$-uniform families in $\mathcal C$.
\end{convention*}

From this perspective,
uniformity conditions can be shown to be equivalent for $\AC^0$.

\begin{theorem}[Barrington, Immerman, Straubing \cite{barrington_uniformity_1990}]\label{thm:uniform-ac0}
	The following classes are equal:
	\begin{enumerate}
		\item logtime-D-uniform $\AC^0$
		\item \FO-D-uniform $\AC^0$
		\item \FO
	\end{enumerate}
\end{theorem}

On the way to our main results we prove the following folklore extension of \cref{thm:uniform-ac0}.

\begin{theorem}\label{thm:regan-vollmer}
	Linear-BD-uniform $\AC^0$ equals logtime-D-uniform $\AC^0$.
\end{theorem}

A stronger version, that logtime-D-uniformity is even
equivalent with linear-BD-uniformity, is stated in~\cite[Section 4]{regan_gap-languages_1997}
and appears to be used implicitly in other work \cite[Section 1.1]{santhanam_uniformity_2014}.
However, this is not something we can prove. The
argument in \cite{regan_gap-languages_1997} seems to silently assume that $\RTIME(n) = \DTIME(n)$ which we
believe is an open problem.

We stress that the results above do not show that the uniformity conditions are equivalent.
In particular, $\RTIME(\log n)$ is strictly weaker than $\FO$.
It merely happens to be the case that the
extra power of $\FO$ for constant-depth circuit family descriptions does not lead to circuit families with more computational power
than any logtime-uniform circuit family.

\subsection{Encoding lists}\label{subsec:lists}
We define an encoding for the lists in
\autoref{def:direct-connection-languages} as a generalization of the
 2-tuple encoding used by Barrington et al.~\cite{barrington_uniformity_1990}.
Each element $x\in 2^*$ is coded as $\delta(|x|)01x$, where $\delta(|x|)$ encodes
the length of $x$ in base $\{00,11\}$ and $01$ acts as a separator.
\begin{definition}\label{def:list}
	We encode $\langle x_0, \dots, x_m \rangle$ as
	\begin{align*}
		\double(|x_0|) 01 x_0
		\double(|x_1|) 01 x_1
		\dots
		\double(|x_m|) 01 x_m
	\end{align*}
	and define projection functions $\pi_i : 2^* \to 2^*$ by
	\begin{align*}
		\pi_i(x) = \begin{cases}
			x_i & \text{if $x$ is of the form $\langle x_0, \dots, x_m \rangle$ with $i \leq m$, and}\\
			0 & \text{otherwise.}
		\end{cases}
	\end{align*}
\end{definition}
For example, $\langle 0, 00, 000 \rangle$ is encoded as
\[
\underbrace{11}_{\delta(|x_0|)}~{01}~
\underbrace{0}_{x_0}~
\underbrace{1100}_{\delta(|x_1|)}~{01}~
\underbrace{00}_{x_1}~
\underbrace{1111}_{\delta(|x_2|)}~{01}~
\underbrace{000}_{x_2}.
\]

The encoding can be worked with in linear time,
in random-access logarithmic time,
and in first order logic with ordering and $\bit$.
In particular,
it can be checked whether a string encodes a list,
the lengths of list items can be extracted,
and the projection functions can be computed,
provided that the relevant item is not too long
(a logtime-machine cannot retrieve items of superlogarithmic length).
Technical details can be found in \autoref{app:proofs}.

\subsection{Parameterized uniformity}
In \cite[Proposition 6]{chen_lower_2019}, Chen and Flum prove a uniform version of $\para\AC^0 = \Para(\AC^0)$.
Their proof carries over to our setting of parameterized linear-BD-uniformity we define in \cref{sec:direct-uniformity} as follows.
In our reformulation we use that parameter functions $\kappa$ are $\FO$-definable.

\begin{theorem}[\cite{chen_lower_2019}]\label{thm:chen}
	Linear-BD-uniform $\para\AC^0$ is equal to
	$\Para(\text{linear-BD-uniform $\AC^0$})$.
\end{theorem}

\section{Direct parameterized uniformity}\label{sec:direct-uniformity}
We first define uniformity notions for parameterized circuit families.
Admissible numberings (\autoref{def:numbering})
carry over to parameterized circuit families in the straightforward way,
and we again assume each parameterized family to come with such a numbering.

\begin{definition}\label{def:direct-connection-languages-para}
	Let $C$ be a parameterized circuit family.
	The \defn{direct connection language} $L_D(C)$
	consists of all strings of the form
	$\langle G, a, p, z, z' \rangle$
	where, writing $n$ for $|z|$ and $k$ for $|z'|$:
	\begin{enumerate}
		\item $G$ is a gate in $C_{n,k}$;
		\item if $p = \epsilon$
			then $a < 3$ and indicates the type of $G$
			(one of \gand/, \gor/, \gnot/);
		\item otherwise, the $p$th gate feeding into $G$ is the gate numbered $a$.
	\end{enumerate}

	The \defn{binary direct connection language} $L_{BD}(C)$
	consists of all strings of the form
	$\langle G, a, p, n, k \rangle$
	satisfying the same conditions,
	now with $n$ and $k$ written in binary.
	That is:
	\[
		L_{BD}(C) = \{\langle G,a,p,\lvert z\rvert,\lvert z'\rvert \rangle : \langle G, a,p,z,z' \rangle \in L_{D}(C)\}.
	\]
\end{definition}

\begin{example}
Consider the language
\begin{align*}
	L = \{ x_0\dots x_{n{-}1} \mid  \forall i \in [1,\lfloor\sqrt{n}\rfloor].\, x_0 = x_{i-1} \text{ and } x_{n{-}i} = x_{n{-}1}\}.
\end{align*}
Figure~\ref{fig:connection-language} shows the circuit $C_{5,2}$ from the
circuit family $C$ that computes the parameterized problem $(L, \kappa)$,
where the parameter function is defined by $\kappa(x) = \lfloor\sqrt{|x|}\rfloor$,
together with some words from the direct connection language $L_D(C$).
\lipicsEnd

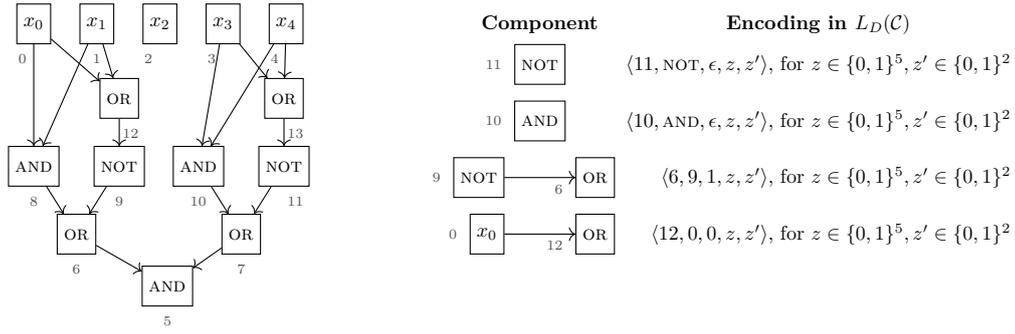
\begin{figure}\centering
	\begin{tikzpicture}[scale=0.75,
		every node/.style={transform shape, draw, minimum height = 2em},
		every label/.style={draw=none, color=lipicsGray, font = \scriptsize},
		label position = below,
		label distance = -1mm,
		grow'=up, <-]

		\path [level 1/.style={sibling distance=15mm,level distance=12mm},
		level 2/.style={sibling distance=10mm,level distance=12mm}]
		node [label={$6$}] (base) {\gor/}
		child { node[label={$8$}] (and) {\gand/}}
		child { node[label={$9$}] (not) {\gnot/}
			child {node[label={[xshift=2mm]$12$}] (or) {\gor/}}
		};

		\path [level 1/.style={sibling distance=15mm, level distance=12mm},
		level 2/.style={sibling distance=10mm, level distance=12mm}]
		node [right=22mm of base, label={$7$}] (base2) {\gor/}
		child { node[label={$10$}] (and2) {\gand/}}
		child { node[label={[xshift=2mm]$11$}] (not2) {\gnot/}
			child {node[label={[xshift=2mm]$13$}] (or2) {\gor/}}
		};

		\node[label={$5$}, below right=2mm and 8mm of base] (output) {\gand/};

		\path
			node [label={[xshift=-2mm]$0$}, above=18mm of and] (x0) {$x_0$}
			node [label={[xshift=0mm]$1$}, right=5mm of x0] (x1) {$x_1$}
			node [label={[xshift=-2mm]$2$}, right=5mm of x1] (x2) {$x_2$}
			node [label={[xshift=-2mm]$3$}, right=5mm of x2] (x3) {$x_3$}
			node [label={[xshift=-2mm]$4$}, right=5mm of x3] (x4) {$x_4$}
			(and) edge (x0)
			(and) edge (x1)
			(or) edge (x0)
			(or) edge (x1) 
			(and2) edge (x3)
			(and2) edge (x4)
			(or2) edge (x3)
			(or2) edge (x4) 
			(output) edge (base)
			(output) edge (base2)
			;

			\node[draw=none, right=30mm of x4] (text) {\textbf{Component}};
			\node[draw=none, right=20mm of text] (text2) {\textbf{Encoding in }$L_D(\mathcal{C})$};
			\node[below=0mm of text, label={[xshift=-2mm]left:$11$}] (ex1) {\gnot/};
			
			\node[draw=none, below right=0mm and 35mm of text2.south, anchor=north east] (word1) {$\langle 11, \gnot/, \epsilon,z,z'  \rangle$, for $z\in \{0,1\}^{5},z'\in \{0,1\}^2$};
			\node[below=10mm of text, label={[xshift=-2mm]left:$10$}] (ex1) {\gand/};
			\node[draw=none, below=10mm of text2.south -| word1.east, anchor=north east, label position=left] {$\langle 10, \gand/, \epsilon,z,z'  \rangle$, for $z\in \{0,1\}^{5},z'\in \{0,1\}^2$};

			\node[below=20mm of text, draw=none] (mid) {};
			\node[left=5mm of mid, label={[xshift=-2mm]left:$9$}] (not9) {\gnot/}; 
			\node[right=5mm of mid, label={[xshift=-2mm,yshift=-2mm]left:$6$}] (or6) {\gor/};
			\path (or6) edge (not9);
			\node[draw=none, below=20mm of text2.south -| word1.east, anchor=north east, label position=left] {$\langle 6, 9, 1,z,z' \rangle$, for $z\in \{0,1\}^{5},z'\in \{0,1\}^2$};

			\node[below=30mm of text, draw=none] (mid2) {};
			\node[left=5mm of mid2, label={[xshift=-2mm]left:$0$}] (in0) {$x_0$}; 
			\node[right=5mm of mid2, label={[xshift=-2mm,yshift=-2mm]left:$12$}] (or11) {\gor/};
			\path (or11) edge (in0);
			\node[draw=none, below=30mm of text2.south -| word1.east, anchor=north east, label position=left] {$\langle 12, 0, 0,z,z' \rangle$, for $z\in \{0,1\}^{5},z'\in \{0,1\}^2$};
		\end{tikzpicture}
	\caption{The circuit $C_{5,2}$ from the family $C$, and associated words $w \in L_D(C)$.}
	\label{fig:connection-language}
\end{figure}
\end{example}

We now define uniformity conditions for parameterized circuit families.
These are based on non-parameterized uniformity conditions.
The difference is that parameterized uniformity conditions allow a precomputation on the same parameter function used for the circuit family.
That is, where before we had $t(n)$ time to decide connections in circuits $C_n$,
we now have $t(n) + f(k)$ time to decide connections in circuits $C_{n,k}$,
for a computable $f$ of choice.

We also adapt \FO-D-uniformity.
Write $\para\FO$ for $\Para(\FO)$.
The motivation for \FO-D-uniformity for parameterized circuit families is that words $\langle G, a, p, z, z' \rangle$
relating to $C_{|z|,|z'|}$ should now be first-order expressible in terms of $|z|$ and $f(|z'|)$,
where $f$ is a computable function of choice.
Again writing $\pi_4$ for the fifth projection function
(so that $\pi_4(\langle G, a, p, z, z' \rangle) = z'$),
this can be expressed as
$(L_D(C), \pi_4)$ being in $\para\FO$.

\begin{definition}
	Let $C$ be a parameterized circuit family.
	We call it
	\begin{description}
		\item [linear-BD-uniform]
			if $(L_{BD}(C), \pi_4)$ is in $\DTIME(n + f(k))$
			for some computable $f$,
			that is, on inputs $w$ of the form $\langle G, a, p, n, k \rangle$
			it takes time $O(|w| + f(k))$;
		\item [logtime-D-uniform]
			if $(L_D(C), \pi_4)$ is in $\RTIME(\log n + f(k))$
			for some computable $f$,
			that is, on inputs $w$ of the form $\langle G, a, p, z, z' \rangle$
			it takes time $O(\log |w| + f(|z'|))$;
		\item [\FO-D-uniform]
			if $(L_D(C), \pi_4)$ is in $\para\FO$.
	\end{description}
\end{definition}

In the following, we sometimes leave the parameter function $\pi_4$ implicit and
e.g. simply state that linear-BD-uniformity is defined as $L_{BD}(C)$ being in $\DTIME(n + f(k))$.

For parameterized circuit families of size $O(f(k)n^c)$,
linear-BD-uniformity implies logtime-D-uniformity,
which in turn implies \FO-D-uniformity.
The former can be shown as follows.

\begin{restatable}{lemma}{linearBDImpliesLogD}\label{lem:linear-bd-implies-log-d}
	Let $C$ be a parameterized circuit family of size bounded by $f(k)n^c$.
	If $C$ is linear-BD-uniform then $C$ is logtime-D-uniform.
\end{restatable}
\begin{proof}[Proof sketch]
	Let $M$ be a machine deciding $(L_{BD}(C), \pi_4)$ in time $O(n + f(k))$,
	and consider the following random-access Turing machine $M'$:
	it rejects all inputs not of the form $\langle G, a, p, z, z' \rangle$
	with $|G|,|a|,|p|$ below $O(\log |z| + f(|z'|))$.
	Otherwise, it writes $\langle G,a,p,|z|,|z'| \rangle$ on a work tape
	and then acts as $M$.
	Now $M'$ witnesses $(L_D(C), \pi_4) \in \RTIME(\log n + g(k))$.
	A full proof can be found in \autoref{app:proofs}.
\end{proof}

Next we prove that logtime-D-uniformity implies \FO-D-uniformity.
It suffices to show the inclusion $\RTIME(\log n + f(k)) \subseteq \para\FO$ for all computable $f$.
We do this by first
showing that $\RTIME(\log n + f(k))$ is contained in linear-BD-uniform $\para\AC^0$,
because this will be used again in \cref{subsec:simulation}.
Through \cref{thm:uniform-ac0,thm:regan-vollmer},
the following result (and proof) can be seen as a parameterized variant of $\RTIME(\log n) \subseteq \FO$
shown in \cite[Proposition 7.1]{barrington_uniformity_1990}.

\begin{restatable}{lemma}{logtimeLinearBDACn}\label{lem:logtime-linear-bd-ac0}
	$\RTIME(\log n + f(k))$ is contained in linear-BD-uniform $\para\AC^0$,
	for all computable $f$.
\end{restatable}
\begin{proof}[Proof sketch]
	A random-access machine $M$ can make at most $O(\log n + f(k))$ queries in $O(\log n + f(k))$ time,
	so a
	constant depth circuit family of size
	$2^{O(\log n + f(k))} = O(g(k)n^c)$
	can nondeterministically branch over all series of query responses,
	verify which is correct (checking all queries in parallel),
	and propagate the output of a simulation of $M$ on the correctly guessed responses to the output gate.
	This can be done linear-BD-uniformly because
	simulating an $O(\log n + f(k))$ time random-access Turing machine
	by answering queries from a given list of guesses
	takes $O(\log n + h(k))$ time on a non-random-access Turing machine.

	A full proof can be found in \autoref{app:proofs}.
\end{proof}

We are now ready to prove that logtime-D-uniformity implies \FO-D-uniformity.
Here we can make use of the work of Barrington et al.~\cite{barrington_uniformity_1990} and that of Chen and Flum~\cite{chen_lower_2019}.

\begin{lemma}\label{lem:logtime-d-implies-fo-d}
	Let $C$ be a parameterized circuit family of size bounded by $f(k)n^c$.
	If $C$ is logtime-D-uniform then $C$ is \FO-D-uniform.
\end{lemma}
\begin{proof}
	Assume $C$ is logtime-D-uniform.
	Then, by definition,
	there is some computable $g$ so that
	$(L_D(C), \pi_4)$ is in $\RTIME(\log n + g(k))$.
	From the inclusions
	\begin{align*}
		\RTIME(\log n + g(k)) & \subseteq \text{linear-BD-uniform $\para\AC^0$} && \text{by \autoref{lem:logtime-linear-bd-ac0}}\\
				      & = \Para(\text{linear-BD-uniform $\AC^0$}) && \text{by \autoref{thm:chen}}\\
				      & = \Para(\FO) && \text{by \cref{thm:regan-vollmer,thm:uniform-ac0}} \\
				      & = \para\FO && \text{by definition}
	\end{align*}
	it follows that $(L_D(C), \pi_4)$ is in \para\FO,
	meaning that $C$ is \FO-D-uniform.
\end{proof}

\section{Uniform \texorpdfstring{$\para\AC^0$ and $\para\AC^{0\up}$}{paraAC0 and paraAC0↑}}
\label{sec:para-ac0-and-para-ac0^}
In the previous section we have seen that,
for parameterized circuit families of size $f(k)n^c$,
linear-BD-uniformity implies logtime-D-uniformity,
which in turn implies \FO-D-uniformity.
It is unknown whether any of these the implications can be reversed.
(See also \cref{ex:uniformity}.)

We can however show implications up-to similar-size circuit equivalence in the reverse direction, by which we mean that
the uniformity conditions do not result in different complexity classes for both $\para\AC^0$ and $\para\AC^{0\up}$.
That is,
for every \FO-D-uniform $O(f(k)n^c)$-sized circuit family of depth $O(d(k))$,
there is a linear-BD-uniform $O(g(k)n^{b})$-sized family of depth $O(d(k))$ that decides the same parameterized problem.

For $\para\AC^0$,
the equality of logtime-D-uniform $\para\AC^0$ and $\FO$-D-uniform $\para\AC^0$
mirrors \cref{thm:uniform-ac0},
and could be proved
using the constant-depth techniques in \cite{barrington_uniformity_1990}.
This does not work for the $f(k)$ depth circuits of $\para\AC^{0\up}$ (the circuits blow up).
We address this by constructing logtime-uniform circuits of $f(k)$ layers that simulate $\FO$-uniform circuits by computing their connection language.
A similar layered approach can be found in \cite{barrington_time_1994},
where they simulate $\FO[t(n)]$ to prove uniformity equivalences for $\AC[t(n)]$.

\begin{restatable}{theorem}{directuniformitytheorem}\label{thm:uniform-para-ac0^}
	The following classes are equal:
	\begin{enumerate}
		\item linear-BD-uniform $\para\AC^{0\up}$\label{thm:uniform-para-ac0^:1},
		\item logtime-D-uniform $\para\AC^{0\up}$\label{thm:uniform-para-ac0^:2}, and
		\item \FO-D-uniform $\para\AC^{0\up}$\label{thm:uniform-para-ac0^:3}.
	\end{enumerate}
	The same holds for $\para\AC^0$.
\end{restatable}

\subparagraph*{Proof outline.}
We will use the following structure,
as visualized in \cref{fig:outline} for $\para\AC^{0}$.
In \cref{sec:direct-uniformity},
we have already proved that
\[
	\text{linear-BD-uniformity} \xRightarrow{\text{Lemma~\ref{lem:linear-bd-implies-log-d}}} \text{logtime-D-uniformity} \xRightarrow{\text{Lemma~\ref{lem:logtime-d-implies-fo-d}}} \FO\text{-D-uniform}
\]
for circuit families of size $O(f(k)n^c)$,
which directly gives the inclusions $(1) \subseteq (2) \subseteq (3)$
for both $\para\AC^0$ and $\para\AC^{0\up}$.
The next step is to prove $(2) \subseteq (1)$,
in particular that linear-BD-uniform $\para\AC^0$ is equal to logtime-D-uniform $\para\AC^0$
(\cref{lem:linear-bd-uniform-para-ac0=logtime-d-uniform-para-ac0}).
We then use this to show $(3) \subseteq (1)$ for both $\para\AC^0$ and $\para\AC^{0\uparrow}$
(\cref{lem:fo-d-implies-linear-bd}).

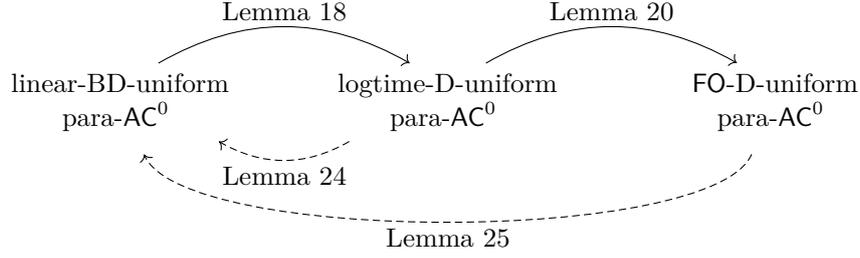
\begin{figure}
\centering
\begin{tikzcd}[math mode=false, text width=9em, align=center, row sep=6em]
	linear-BD-uniform $\para\AC^{0\ }$
		\arrow[r, bend left, "\cref{lem:linear-bd-implies-log-d}" above]
	& logtime-D-uniform $\para\AC^{0\ }$
		\arrow[r, bend left, "\cref{lem:logtime-d-implies-fo-d}" above]
		\arrow[l, dashed, bend left, "\cref{lem:linear-bd-uniform-para-ac0=logtime-d-uniform-para-ac0}"
			{name = first, below}]
	& \FO-D-uniform $\para\AC^{0\ }$
		\arrow[ll, dashed, controls={+(-0.5,-1.2) and +(0.5,-1.2)}, shift left=2, "\cref{lem:fo-d-implies-linear-bd}"
			{name = second}]
\end{tikzcd}
\caption{Visualization of the steps in the proof of \cref{thm:uniform-para-ac0^} for $\para\AC^0$.
Solid arrows are direct implications on the level of circuit families,
dashed arrows indicate implications up-to similar-size circuit equivalence.
\cref{lem:linear-bd-uniform-para-ac0=logtime-d-uniform-para-ac0} is used to prove \cref{lem:fo-d-implies-linear-bd}.
The proof for $\para\AC^{0\up}$ follows the same structure.
}
\label{fig:outline}
\end{figure}

The new inclusions will be proved by simulation.
We will show that there are linear-BD-uniform constant-depth circuit families
for deciding the connection language $L_D(C)$
of logtime-D-uniform and \FO-D-uniform circuit families $C$ of size $O(f(k)n^c)$.
We then compose them into larger linear-BD-uniform circuit families to simulate $C$ itself.

To prove that these compositions are indeed linear-BD-uniform,
we first show that the class of linear-BD-uniform circuit families is in some sense closed under substitution.

\subsection{Substitution}\label{subsec:substitution}
A main ingredient of our construction is that within linear-BD-uniformity we are able to
\emph{compose} circuit families.
That is, given uniform circuit families $A$ and $B$,
we can uniformly replace some of the gates in $A$ by circuits from $B$.
This can be formalized as follows.

\begin{definition}\label{def:substitution}
	Let $A,B$ be parameterized circuit families,
	let $M(G,n,k): \N^3 \to 2$
	be a function marking internal gates in $A_{n,k}$,
	and let $\fanin(G,n,k)$ be the fan-in of gate $G$ in $A_{n,k}$.
The \defn{substitution} $A[B/M]$
	is the parameterized circuit family where $A[B/M]_{n,k}$
	consists of subcircuits $C(G)$ for each gate $G$ in $A_{n,k}$ so that
	\begin{enumerate}
		\item $C(G) = G$ if $M(G,n,k) = 0$,
		\item $C(G) = B_{\fanin(G,n,k),k}$ if $M(G,n,k) = 1$
			(input and output gates become unary \gor/ gates),
		\item and if $H$ is the $i$th input of $G$ in $A_{n,k}$,
			then $C(H)$ is the $i$th input of $C(G)$ in $A[B/M]_{n,k}$.
	\end{enumerate}
	That is, $A[B/M]_{n,k}$ is obtained from $A$ by replacing all marked gates with circuits of $B$.
\end{definition}

It can be shown that
$A[B/M]$ is linear-BD-uniform if both $A$ and $B$ are,
and in addition $M$ and $\fanin$ are reasonably computable.

\begin{restatable}{lemma}{substitutionLemma}\label{lem:substitution}
	Let $A,B$ be linear-BD-uniform parameterized circuit families
	of size $O(f(k)n^c)$ for some computable $f$.
	Let $M(G,n,k): \N^3 \to 2$ be a function marking internal gates in $A_{n,k}$,
	and let $\fanin(G,n,k)$ be the fan-in of gate $G$ in $A_{n,k}$.

	Now $A[B/M]$ is linear-BD-uniform, provided that there is a computable $h$ such that:
	\begin{enumerate}[(1)]
		\item $M(G,n,k)$ is computable in time $|G|+|n|+h(k)$, and
		\item if $M(G,n,k) = 1$ then $\fanin(G,n,k)$ is computable in time $|G|+|n|+h(k)$.
	\end{enumerate}
\end{restatable}
\begin{proof}[Proof sketch]
	It suffices to define an admissible numbering for $D= A[B/M]$,
	and show that,
	under this numbering,
	$(L_{BD}(D), \pi_4)$ is in $\DTIME(n + f(k))$ for some computable $f$.

	Give unmarked internal gates $G$ in $A_{n,k}$ number $\langle 0, G \rangle$ in $D_{n,k}$.
	If $G$ is a marked gate and $G'$ is a gate in the circuit $B_{\fanin(G,n,k),k}$ replacing it,
	number it $\langle G, G' \rangle$ in $D_{n,k}$.
	This numbering based on admissible numberings of $A$ and $B$ leads to an admissible numbering of $D$ (cf. \cref{def:numbering}).
	An $O(n + f(k))$-time algorithm deciding $(L_{BD}(D), \pi_4)$ is then obtained
	from $M$, $\fanin$, and $O(n + g(k))$-time algorithms for $L_{BD}(A)$ and $L_{BD}(B)$
	witnessing the linear-BD-uniformity of $A$ and $B$.
	A full proof can be found in \autoref{app:proofs}.
\end{proof}

We do not know whether the logtime-D-uniform (parameterized) circuit families are closed under substitution,
for reasonable assumptions on $M$ and $\fanin$.
It appears that two
random-access Turing machines $P,Q$ are not easily composed
without a runtime overhead.
For example, computing $P(x_r)$ on input $x = \langle x_l, x_r \rangle$
seems to require offsetting each query that $P$ makes by an $O(|x_l|)$ offset
which, if done naively, adds an $O(\log |x_l|)$ overhead to every query.
These
naive combinations of two logtime random-access Turing machines
can result in $\Omega((\log n)^2)$-time random-access Turing machines,
instead of logtime.

\subsection{Simulation}\label{subsec:simulation}
We are now ready to prove \cref{thm:uniform-para-ac0^} according to the outline given at the start of the section.
The first step is the indirect inclusion from logtime-D-uniformity to linear-BD-uniformity.

\begin{lemma}\label{lem:linear-bd-uniform-para-ac0=logtime-d-uniform-para-ac0}
Let $C$ be a parameterized circuit family of size $O(f(k)n^c)$ and depth bounded by $d(k)$,
	with $f,d$ computable.
	If $C$ is logtime-D-uniform,
	there is an equivalent linear-BD-uniform family $D$ of size $O(f'(k)n^{c'})$ and depth $O(d(k))$,
	with $f'$ computable.
\end{lemma}
\begin{proof}
	Let $a$ be a constant and $h$ a computable function so that for all $n,k$
	we have that $N(n,k) = 2^{a\log n + h(k)}$ is greater than any gate number in $C_{n,k}$ (cf. \autoref{def:numbering}).
	Then the binary representation of $N(n,k)$ is computable in $O(\log n + g(k))$ time without random access for some computable $g$.

	We describe a linear-BD-uniform circuit family $D$ equivalent with $C$.
	Circuit $D_{n,k}$ consists of its input gates,
	$d(k)$ layers of $N(n,k)$ subcircuits we call \defn{simgates},
	and one final subcircuit that forwards the correct simgate to the output gate.
	Every simgate of layer $m{+}1$ takes all simgates of layer $m$ as input,
	together with the input gates.
	We number the $q$th simgate in layer $m$ with $\langle m, q \rangle$,
	and the idea is that simgate $\langle m, q \rangle$ in $D_{n,k}$ approaches (monotonically)
	gate $q$ in $C_{n,k}$,
	so that, by layer $d(k)$, simgate $\langle d(k), q \rangle$ correctly simulates gate $q$ in $C_{n,k}$.

	\begin{figure}[b!]\centering
\begin{tikzpicture}[scale=0.75,
			every node/.style={transform shape, draw, minimum height = 2em},
			every label/.style={draw=none, color=lipicsGray, font = \scriptsize},
			label position = below,
			label distance = -1mm,
			phi/.style={rounded corners, dashed}, grow'=up, <-]
			\path [level 1/.style={sibling distance=35mm},
				level 2/.style={sibling distance=25mm}]
				node [label={$\langle m, q, 0 \rangle$}] (base) {\gor/}
				child { node[label={$\langle m, q, 1 \rangle$}] (simOR) {\gand/}
					child {node[phi, label={$\langle m, q, 4 \rangle$}] {$q$ is \gor/?}}
					child[missing]
				}
				child { node[label={$\langle m, q, 2 \rangle$}] {\gand/}
					child { node[label={[xshift=2mm]$\langle m, q, 5 \rangle$}] (simNOT) {\gnot/}}
					child {node[phi, label={[xshift=-2mm]$\langle m, q, 6 \rangle$}] (q_is_not) {$q$ is \gnot/?}}
				}
				child { node[label={$\langle m, q, 3 \rangle$}] (simAND) {\gand/}
					child[missing]
					child {node[phi, label={[xshift=-2mm]$\langle m, q, 7 \rangle$}] {$q$ is \gand/?}}
				};

\path
				node [label={$\langle m, q, 8 \rangle$}, above=10mm of simNOT] (OR) {\gor/}
				edge [->] (simOR.north)
				edge [->] (simNOT.north)
				node [label={$\langle m, q, 9 \rangle$}] at (q_is_not |- OR) (AND) {\gand/}
				edge [->] (simAND.north);

			\path [sibling distance=20mm]
				node [above=of OR, xshift=-2.5cm, label={$\langle m, q, 10, i \rangle$}] (I0) {\gand/}
				child[missing]
				child {node (input_n) {input $i$}}
				child {node[phi, label=above:{$\langle m, q, 14, i \rangle$}] (q_takes_n) {$q$ takes $i$?}}
				;

\path
				node [label={$\langle m, q, 11, i \rangle$}] at (q_takes_n |- I0) (I0o) {\gor/}
				edge (input_n)
				edge (q_takes_n);

			\path [sibling distance=20mm]
				node [right=5cm of I0, label={$\langle m, q, 12, p \rangle$}] (p) {\gand/}
				child[missing]
				child {node (gate_p) {$\langle m - 1, p \rangle$}}
				child {node[phi, label=above:{$\langle m, q, 15, p \rangle$}] (q_takes_p) {$q$ takes $p$?}}
				;

\path
				node [label={$\langle m, q, 13, p \rangle$}] at (q_takes_p |- p) (po) {\gor/}
				edge (gate_p)
				edge (q_takes_p);

\node[draw=none, below left=0.5cm and 0.25cm of input_n] (dotsl) {$\dots$};
			\node[draw=none, below right=0.5cm and 0.25cm of q_takes_n] (dotsl) {$\dots$};
			\node[draw=none, below left=0.5cm and 0.25cm of gate_p] (dotsl) {$\dots$};
			\node[draw=none, below right=0.5cm and 0.25cm of q_takes_p] (dotsl) {$\dots$};

\path ([xshift=-0.1cm]OR.north) edge (I0.south)
				([xshift=0.1cm]OR.north) edge (p.south)
				([xshift=-0.1cm]AND.north) edge (I0o.south)
				([xshift=0.1cm]AND.north) edge (po.south);
		\end{tikzpicture}
		\caption{
			The internals of simgate $\langle m, q \rangle$ simulating gate $q$ of $C_{n,k}$,
			annotated with an admissible numbering.
			It takes the $n$ input gates and the $N(n,k)$ simgates of the previous layer as input.
			The boxes with rounded corners represent subcircuits that
			determine how to simulate $q$ by
			querying $L_D(C)$.
			If there is no gate numbered $q$ in $C_{n,k}$,
			then the output of the simgate is 0.
		}
		\label{fig:simgate}
	\end{figure}
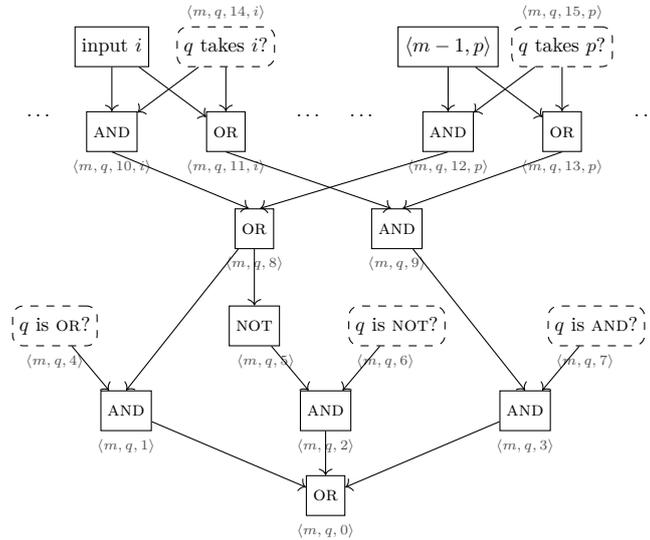

	See \autoref{fig:simgate} for the internals of a simgate.
	In order to approach gate $q$ of $C_{n,k}$,
	simgate $\langle m, q \rangle$ queries $L_D(C)$
	to select the simgates $\langle m{-}1, p \rangle$ where $p$ feeds into $q$ in $C_{n,k}$,
	and the input gates $m$ where input gate $m$ feeds into $q$ in $C_{n,k}$.
	It also queries $L_D(C)$ to determine whether $q$ is an \gand/, \gor/, or \gnot/ gate,
	and simulates that action on its selected inputs.
	(By our construction in \cref{fig:simgate},
	if $q < N(n,k)$ is not a gate in $C_{n,k}$ then simgate $\langle m, q \rangle$ outputs 0,
	but it has no influence on the output of $D_{n,k}$, because it is never selected.)

	Assume for now that components of simgates can directly answer membership
	queries for $L_D(C)$.
In this case the construction shows that they select their input as described above,
	and perform the action of $q$ (if $q$ is a gate in $C_{n,k}$).
Then by induction, for every gate $q$ in $C_{n,k}$,
	the simgate $\langle d(k), q \rangle$ in the final layer will have the same output
	as $q$.

	To see this, first observe that by our assumption,
	the output of $\langle m+1,q\rangle$ is the same as that of $q$
	if for all internal gates $p$ that feed into $q$, the output of $\langle m,p \rangle$ is equal to that of $p$.
Define the \defn{level} of a gate to be the length of the longest path from it to an input gate.
	The output of simgates $\langle m, q \rangle$ is the same as that of $q$
	if the level of $q$ is below $m$.
	This can be shown by induction based on the following:
	(1) If $q$ has level 1 in $C_{n,k}$,
	then $q$ only takes input gates,
	and by our observation simgates $\langle m, q \rangle$ have the same output as $q$ for all $m \geq 1$.
	(2) If $q$ has level $m+1$ and simgates $\langle m',p \rangle$ have the same output as $p$ for all $p$ of lower level and $m' \geq m$,
	then by our observation $\langle m'', q\rangle$ has the same output as $q$ for all $m'' \geq m + 1$.

	Finally, to make $D_{n,k}$ behave as desired,
	it should feed into the output gate the output of simgate $\langle d(k), q \rangle$
	for the gate $q$ that feeds into the output gate in $C_{n,k}$.
	This is the role of the final subcircuit mentioned at the start of the proof.
	It takes all simgates in the final layer as input,
	and queries $L_D(C)$ to propagate the output of the correct $\langle d(k), q \rangle$.

	The circuits $D_{n,k}$ have very regular structure.
	Each $D_{n,k}$ is essentially a matrix of $d(k)$ layers of $N(n,k)$ simgates,
	and each simgate gets all the simgates of the previous layer as input together with the input gates.
	Apart from the rounded query blocks,
	the internals of the simgates are also easily seen to be linear-BD-uniform,
	using for example the admissible numbering in \autoref{fig:simgate}.
	It remains to show how the queries in the simgates are implemented.
	\proofsubparagraph{Simgate queries.}
	The $L_D(C)$ queries represented by the dashed query blocks in \autoref{fig:simgate} can be implemented
	by plugging in linear-BD-uniform circuits for $L_D(C)$.
	A constant-depth linear-BD-uniform circuit family $E$ deciding $L_D(C)$ exists by \autoref{lem:logtime-linear-bd-ac0},
	and its circuits can be substituted in linear-BD-uniformly by \autoref{lem:substitution}.
It is left to prove that the circuits can be given the appropriate query inputs
	in a linear-BD-uniform way.

	For now, we work with the circuit family $D'$ in which the circuits of $E$ have not yet been substituted.
	Here gates $\langle m, q, 4 \rangle$ (cf. \cref{fig:simgate}) are just \gand/ gates
	that will be replaced with circuits from $E$ later to obtain $D$ from $D'$.
	We add constant circuits $V_0$ and $V_1$ to every $D'_{n,k}$, where
$V_0$ has the constant output 0, and $V_1$ is constant 1.
The circuit of $E$ replacing $\langle m, q, 4 \rangle$ should decide whether gate $q$ in $C_{n,k}$ is an \gor/ gate.
	This can be achieved by giving gate $\langle m, q, 4 \rangle$ the input
	$I = \langle q, 1, \epsilon, 1^n, 1^k \rangle$.
	That is, in terms of connection languages,
	we want $w = \langle \langle m, q, 4 \rangle, V_b, i, n, k \rangle \in L_{BD}(D')$
	if and only if the $i$th bit of $I$ is $b$.
	It can be verified that, given such $w$,
	the $i$th bit of $I$ can be decided in time $|w| + g(k)$ for some computable $g$,
	and so these queries can be answered in time $|w| + g(k)$ as required for linear-BD-uniformity.

	This addresses the type-queries in the simgates.
	The other queries, those asking whether (input) gate $p$ feeds into $q$ in $C_{n,k}$,
	are handled similarly,
	with a small complication: 
	checking whether $p$ feeds into $q$ can be reduced to checking whether $p$ is the $i$th input of $q$
	for any $i \leq N(n,k)$,
	so the dashed blocks in \cref{fig:simgate} corresponding to this type of query
	are actually \gor/ gates over $N(n,k)$ gates
	that will be substituted by circuits of $E$,
	where the $i$th gate will be given an $L_D(C)$-query to determine whether $p$ is the $i$th input of $q$.

In summary: the queries in the simgates are implemented
	by hard-wiring $L_D(C)$-queries into the query blocks in a linear-BD-uniform manner,
	and then substituting in a linear-BD-uniform circuit family $E$ that decides $L_D(C)$ to answer them.
	The resulting circuit family $D$ is linear-BD-uniform by \cref{lem:substitution},
	and is equivalent with $C$ by design.
\end{proof}

\cref{thm:regan-vollmer} (mentioned in \cref{sec:uniform-ac0}) can be proved similarly.
Using simgates, we can now also prove the following.

\begin{lemma}\label{lem:fo-d-implies-linear-bd}
	Let $C$ be a parameterized circuit family of size $O(f(k)n^c)$ and depth bounded by $d(k)$,
	with $f,d$ computable.
	If $C$ is \FO-D-uniform,
	there is an equivalent linear-BD-uniform family $D$ of size $O(f'(k)n^{c'})$ and depth $O(d(k))$,
	with $f'$ computable.
\end{lemma}
\begin{proof}
	By the \FO-D-uniformity of $C$, we have $(L_D(C), \pi_4) \in \Para(\FO)$.
	\autoref{thm:uniform-ac0} gives that $(L_D(C), \pi_4)$ is in $\Para(\text{logtime-D-uniform $\AC^0$})$,
	which by \autoref{thm:chen} is equal to logtime-D-uniform $\para\AC^0$.
	Hence, by \autoref{lem:linear-bd-uniform-para-ac0=logtime-d-uniform-para-ac0},
	$(L_D(C), \pi_4)$ is in linear-BD-uniform $\para\AC^0$.
	The rest of the proof is now the same as for \cref{lem:linear-bd-uniform-para-ac0=logtime-d-uniform-para-ac0},
	using a simgate-construction of $d(k)$ layers,
	and answering $L_D(C)$-queries in the simgates by substituting in a family witnessing
	that $(L_D(C), \pi_4)$ is in linear-BD-uniform $\para\AC^0$.
\end{proof}

We have now proved our main result.

\begin{proof}[Proof of Theorem~\ref{thm:uniform-para-ac0^}]
	For both $\para\AC^0$ and $\para\AC^{0\up}$,
	Item (\ref{thm:uniform-para-ac0^:1}) implies Item (\ref{thm:uniform-para-ac0^:2}) by \autoref{lem:linear-bd-implies-log-d}, and (\ref{thm:uniform-para-ac0^:2}) implies (\ref{thm:uniform-para-ac0^:3}) by \autoref{lem:logtime-d-implies-fo-d}. Finally, Item~(\ref{thm:uniform-para-ac0^:3}) implies Item~(\ref{thm:uniform-para-ac0^:1}) by \autoref{lem:fo-d-implies-linear-bd}.
	(For $\para\AC^0$, the last implication is obtained from \cref{lem:fo-d-implies-linear-bd} by setting $d(k) = O(1)$.)
\end{proof}

We can also observe the following,
which suggests that $\FO$-D-uniformity is not too coarse for $\para\AC^{0\up}$.
(This mirrors that (\FO-uniform $\AC^0$)-uniform $\AC^0$ equals $\FO$-uniform $\AC^0$.)

\begin{corollary}
	(\FO-D-uniform $\para\AC^{0\up}$)-uniform $\para\AC^{0\up}$ = \FO-D-uniform $\para\AC^{0\up}$.
\end{corollary}
\begin{proof}[Proof sketch]
	If $(L_D(C), \pi_4)$ is in $\FO$-uniform $\para\AC^{0\up}$,
	then by \cref{thm:uniform-para-ac0^} it lies in linear-BD-uniform $\para\AC^{0\up}$.
	Follow the proof of \cref{lem:linear-bd-uniform-para-ac0=logtime-d-uniform-para-ac0},
	now using linear-BD-uniform $\para\AC^{0\up}$-families for the query blocks.
\end{proof}

\section{Extended uniformity}\label{sec:extended-uniformity}
Thus far we considered direct uniformity notions,
where the connection language only contains information on directly related gates.
Extended uniformity notions, where the connection language also describes paths between gates, are of interest because logtime-E-uniform circuit classes
such as logtime-E-uniform $\AC^0$ and $\NC^i$ are equal to natural alternating Turing machine classes \cite{vollmer_introduction_1999}.
It is unknown whether logtime-D-uniformity and logtime-E-uniformity are equivalent in general,
but it has been shown that logtime-D-uniform $\AC^0$ is equal to logtime-E-uniform $\AC^0$
\cite[Theorem 4.31]{vollmer_introduction_1999}.
We will show that this also holds for $\para\AC^0$ and, more interestingly, for $\para\AC^{0\up}$.

\begin{definition}
	Let $C$ be a circuit.
	A \defn{path} from gate $G$ to gate $a$
	is a list $p = \langle p_1, p_2, \dots, p_m \rangle$
	such that there exist $G_0, \dots, G_m, G_{m+1}$ in $C$
	where for each $i$, $G_{i+1}$ is the $p_{i+1}$th predecessor (input) of $G_i$,
	and where $G_0 = G$ and $G_{m+1} = a$.
\end{definition}

In the extended connection language,
the $p$ in $\langle G, a, p, z, z' \rangle$
are paths from $G$ to $a$ \cite{vollmer_introduction_1999}.

\begin{definition}\label{def:extended-connection-language-para}
	Let $C$ be a parameterized circuit family.
	The \defn{extended connection language} $L_E(C)$
	consists of all strings of the form
	$\langle G, a, p, z, z' \rangle$
	where, writing $n$ for $|z|$ and $k$ for $z'$:
	\begin{enumerate}
		\item $G$ is a gate in $C_{n,k}$;
		\item if $p = \epsilon$
			then $a < 3$ indicates the type of $G$
			(one of \gand/, \gor/, \gnot/);
		\item otherwise, $p$ is a path from $G$ to $a$.
	\end{enumerate}
\end{definition}

In \cite[p. 123]{vollmer_introduction_1999},
Vollmer defines logtime-E-uniformity for
non-parameterized polynomial-size constant-depth circuit families $C$
as $L_E(C)$ being in $\RTIME(\log n)$.
We adapt this to $O(d(k))$-depth parameterized circuit families as follows.

\begin{definition}
	A parameterized circuit family of size $O(f(k)n^c)$ and depth $O(f(k))$
	is \defn{logtime-E-uniform} if $(L_E(C), \pi_4)$ is $\RTIME(\log n + g(k))$ for a computable $g$.
\end{definition}

For $\para\AC^{0\up}$,
this definition is quite strict:
there are logtime-D-uniform circuit families of size $O(f(k)n^c)$ and depth $d(k)$
that contain valid paths $p$ of $d(k)$ steps,
where every step has size $\Omega(\log n)$.
The encoding of such paths $p$ has length $\Omega(d(k)\log n)$.
No $\RTIME(\log n + g(k))$ machine in \cref{def:extended-connection-language-para}
has time to read the complete path $p$.
Nevertheless,
these circuit families still have logtime-E-uniform representatives for $\para\AC^{0\up}$.

\begin{restatable}{lemma}{linearBDImpliesLogtimeE}\label{lem:linear-BD-implies-logtime-E}
	Let $A$ be a parameterized circuit family of size $O(f(k)n^c)$ and depth bounded by $d(k)$,
	with $f,d$ computable.
	If $A$ is linear-BD-uniform,
	there is an equivalent logtime-E-uniform circuit family $C$
	of size $O(f'(k)n^{c'})$ and depth $O(d(k))$,
	with $f'$ computable.
\end{restatable}
\begin{proof}[Proof sketch]
	As in the proofs of \cref{lem:linear-bd-uniform-para-ac0=logtime-d-uniform-para-ac0,lem:fo-d-implies-linear-bd},
	$C$ will be made up of $d(k)$ layers of simgates.
	The main idea is that only steps between layers are difficult,
	and that in order to decide to which simgate a path leads,
	we only need to read the last inter-layer step.
	We achieve this by making the number of simgates per layer exactly $2^{L(n,k)}-1$ for some appropriate $L(n,k)$.
	Then, if $C$ were to consist only of these simgates,
	and the simgates were actual gates instead of subcircuits,
	it could be checked whether a path $p$ moves from simgate $\langle m, q \rangle$ to $\langle m', q' \rangle$
	by making sure that
	\begin{enumerate}
		\item each step in $p$ has length at most $L(n,k)$,
		\item $d(k) \geq m \geq m' \geq 1$ and $p$ contains $m-m'$ steps, and
		\item the last step in $p$ is $q'$.
	\end{enumerate}
	So, only the final step in $p$ has to be read.
	For the others it suffices to read their lengths.
	From this it can be shown that the whole procedure can be performed in $O(h(k) + \log n)$ time for some computable $h$.
	\autoref{app:proofs} contains a full proof.
\end{proof}

Because logtime-E-uniform families of size $O(f(k)n^c)$ and depth $f(k)$ are by definition also logtime-D-uniform
and hence (by \autoref{thm:uniform-para-ac0^})
equivalent to a linear-BD-uniform family of similar dimensions,
we may now conclude the following.

\begin{corollary}
	Logtime-E-uniform $\para\AC^{0\up}$ equals logtime-D-uniform $\para\AC^{0\up}$,
	and logtime-E-uniform $\para\AC^{0}$ equals logtime-D-uniform $\para\AC^{0}$.
\end{corollary}

We can also define \FO-E-uniformity as $(L_E(C), \pi_4)$ being $\para\FO$.
By the same proof as \cref{lem:logtime-d-implies-fo-d},
it can then be shown that logtime-E-uniformity implies \FO-E-uniformity
for parameterized circuit families of size $O(f(k)n^c)$,
and thus
(because $\FO$-E-uniformity implies \FO-D-uniformity,
and by then applying \cref{thm:uniform-para-ac0^})
\FO-E-uniform $\para\AC^{0\up}$ is equivalent with logtime-D-uniform $\para\AC^{0\up}$.
The same also holds for $\para\AC^0$.

\section{Relation to descriptive complexity}\label{sec:descriptive}
In this section we discuss a descriptive classification of logtime-uniform $\para\AC^{0\up}$
mirroring the equivalence of $\FO$ and logtime-uniform $\AC^0$ in \cref{thm:uniform-ac0}.

For non-parameterized circuits,
Barrington and Immerman~\cite{barrington_time_1994} prove that for sufficiently constructible polynomially bounded functions $t$,
logtime-D-uniform and $\FO$-D-uniform $\AC[t(n)]$ are equal to the class $\FO[t(n)]$ of languages decided by iterated first order sentences.

A similar observation can be made in the parameterized setting.
Define $\para\AC[f(k)]$ in terms of parameterized circuit families of size $O(f(k)n^c)$ and depth $O(f(k))$,
so that $\para\AC^{0\up}$ is the union of $\para\AC[f(k)]$ over the computable functions $f$.
Based on the definition of $\FO[t(n)]$ in \cite[Definition 4.24]{immerman_descriptive_1999},
we define $\para\FO[f(k)]$ as follows.

Write $(\forall x. \phi)\psi$ for $(\forall x)(\phi(x) \implies \psi)$
and $(\exists x. \phi)\psi$ for $(\exists x)(\phi(x) \land \psi)$.
A \defn{quantifier block} is a partial formula of the form $(Q_0 x_0. M_0)(Q_1 x_1. M_1) \dots (Q_m x_m. M_m)$,
where the $Q_i$ are quantifiers, the $x_i$ are (not necessarily distinct) variable letters, and the $M_i$ are quantifier-free formulas in the language of $\FO$.

\begin{definition}\label{def:FO[f(k)]}
	A parameterized problem $(Q, \kappa)$ is in $\para\FO[f(k)]$ if there is
	a computable $g$,
	a quantifier block $M$,
	a quantifier-free formula $\psi$,
	and a set of constants $c$
	such that
	\begin{align*}
		y \models M^{f(\kappa(x))}\psi(c) \Longleftrightarrow y \in \{ \langle x, g(\kappa(x)) \rangle : x \in Q \}
	\end{align*}
	where $M^{f(\kappa(x))}\psi(c)$ stands for the sentence obtained
	by concatenating $f(\kappa(x))$ copies of $M$ with $\psi$,
	and then substituting $c_i$ for free occurrences of variables $x_i$.
\end{definition}

Intuitively,
a parameterized problem $(Q, \kappa)$ is in $\para\FO[f(k)]$ if it is first-order describable in $|x|,f(\kappa(x))$ and iterations of length $f(\kappa(x))$,
that is, if $x \in Q \Longleftrightarrow x \models M^{f(\kappa(x))}\psi(c)$ holds.
That this corresponds with \cref{def:FO[f(k)]} can be made precise by adding constant symbols for $|x|$ and $f(\kappa(x))$ to the language,
and extending the domain of $x$ from $|x|$ to $\max(|x|, f(\kappa(x)))$.

Following the proof strategies in \cite{barrington_time_1994} (see also \cite{immerman_descriptive_1999}) it follow that, writing $\C$ for the computable functions $\N \to \N$:
\begin{align*}
	\text{\FO-D-uniform $\para\AC^{0\up}$}
	= \bigcup_{f \in \C} \text{\FO-uniform $\para\AC[f(k)]$}
	= \bigcup_{f \in \C} \para\FO[f(k)],
\end{align*}
thus linking logtime-D-uniform $\para\AC^{0\up}$ to a descriptive complexity class.
Following \cite{barrington_time_1994,immerman_descriptive_1999} further leads to an alternative, descriptive complexity-based proof for \cref{thm:uniform-para-ac0^}.

\section{Conclusion}
We have shown that, for $\para\AC^0$ and $\para\AC^{0\up}$,
various uniformity conditions found in ordinary circuit complexity are again equivalent in the sense that the respective induced complexity classes are equal.
This can be used to repair uniformity gaps in the literature.

For example,
while our result does not directly demonstrate that the circuit family  
mentioned in \cite[Theorem 3.2]{bannach_fast_2015}
is logtime-D-uniform as claimed, it does prove that there is a logtime-D-uniform circuit family equivalent with it.
Since these circuits are described in terms of additions, multiplications, and modulo operations of numbers
below a polynomial in $n$ and parameter $k$,
it is straightforward to show that the circuit family is \FO-D-uniform, and so, by \cref{thm:uniform-para-ac0^},
there is an equivalent logtime-D-uniform family.

We obtain our equivalences via the stricter parameterized linear-uniformity condition,
proving a substitution lemma
and then explicitly constructing layered linear-uniform families for given $\FO$-uniform families.
In a sense this can be seen as a direct, circuit-oriented variant of the approach in \cite{barrington_time_1994}.
Here they prove the equivalence of (direct, non-parameterized) uniformity conditions
for depth-$t(n)$ circuits via the descriptive complexity class $\FO[t(n)]$,
constructing layered logtime-uniform circuits deciding $\FO[t(n)]$-formulas.

Our results show the equivalence between the uniformity conditions not on the circuit family level
but only on the complexity class level.
This suffices for theoretical applications like \cite[Theorem 3.2]{bannach_fast_2015} mentioned above,
but is not fine-grained enough for questions of work efficiency.
For example, the proof of \cref{thm:uniform-para-ac0^} gives
for an \FO-D-uniform circuit family of size $f(k)n^c$
a logtime-D-uniform circuit family of size greater than $f(k)n^{3c+1}$.
This gap limits the study of open problems in fine-grained parallel complexity.
It would therefore be interesting to study the inherent cost of the conversion,
and to also consider which uniformity condition is the most natural one from a (parameterized) parallel complexity perspective.

\bibliographystyle{plainurl}
\bibliography{references}

\newpage
\appendix

\section{Beyond \texorpdfstring{$\FO$}{FO}-computable parameters}\label{sec:kappa-fo}
So far we have worked with the definition of $\para\FO$ in the spirit of Flum and Grohe (cf. \cref{def:Para}),
where
$(Q, \kappa)$ is in $\para\FO$ if and only if
$\{ \langle x, f(\kappa(x)) \rangle : x \in Q \}$ is in $\FO$.
Intuitively, this definition leaves it up to the $\FO$-formula
to decide for given input $\langle x, y \rangle$
whether $y$ codes the expected precomputation $f(\kappa(x))$.
Assuming $\kappa$ is $\FO$-uniform, this can be done:
for any computable $f$ there is a computable $g$ and a formula $\phi$
so that $\pi_0(g(\kappa(x))) = f(\kappa(x))$
and $\phi$ decides $y = g(\kappa(x))$ from $y$ and $\kappa(x)$.
(For example, $f$ can code a computation of $g(\kappa(x))$.)
This is why it is helpful to assume that $\kappa$ is $\FO$-computable.

The assumption on $\kappa$ can be weakened by slightly modifying the definition of $\para\FO$,
similar to the modification of $\Para(\AC^0)$ in \cite[Proposition 6(ii)]{chen_lower_2019}.
Say that
$(Q, \kappa)$ is in $\p\FO$ if and only if
there is an $\FO$-formula $\phi$ with
$x \in Q \Longleftrightarrow \langle x, f(\kappa(x)) \rangle \models \phi$.
Then the results in this paper hold for $\p\FO$ in place of $\para\FO$,
not just for $\FO$-computable $\kappa$ but
for all computable $\kappa$ with $|\kappa(x)| \leq O(n^{O(1)})$.
An alternative definition for $\pFO$
(equivalent up to the choice of $f$)
would be to add constant symbols for $|x|$ and $f(\kappa(x))$ to the language,
and extending the domain of inputs $x$ to $\max(|x|, f(\kappa(x)))$.

Technical details can be found in Appendix~\ref{app:pfo}.

\section{Technical appendix}\label{app:proofs}
\subsection{List encoding}
We verify that the encoding in \cref{def:list} can be worked with in linear time,
in random-access logarithmic time,
and in first order logic with ordering and $\bit$.
In particular,
that it can be checked whether a string encodes a list,
that the lengths of list items can be extracted,
and that the projection functions can be computed,
provided that the relevant item is not too long
(a logtime-machine cannot retrieve items of superlogarithmic length).

We focus on the linear and logarithmic time machines.
For \FO,
the claims are relatively easy to verify directly,
or using $\RTIME(\log n) \subseteq \FO$
\cite[Proposition 7.1]{barrington_uniformity_1990}.

To show that the linear-time machines can work with the encoding,
it suffices (roughly) to show that,
given some number $M$ in binary,
say the length $|x_i|$ of a list item,
a linear-time machine can count down to 0 in order to copy exactly the $|x_i|$ bits of $x_i$,
or to move exactly $|x_i| + 2$ bits in order to reach the start of the encoding of $|x_{i+1}|$.

The strategy for logtime-machines is similar,
the difference being that they do not move their reading head but instead offset
the address written (in binary) on their query tape.
Barrington et al. touch on this in \cite[Lemma 6.1]{barrington_uniformity_1990}
(our list encoding is similar to theirs).
We stress here that querying a block of $M = O(\log n)$ bits does not take too much time:
counting in binary from $N$ to $M$ on the query tape,
querying in-between,
takes $O(M)$ time.

In summary,
that both machines
can work with the list encoding follows from the following.

\begin{lemma}\label{lem:counting}
	A multitape machine can count in binary from $N$ to $N+M$ in $O(M)$ time.
\end{lemma}
\begin{proof}[Proof sketch]
	It is straightforward to show that the total number of carry bits
	forms a geometric series
	(the least significant bit has a carry every two successor operations,
	the one after that every four, etc.),
	and so it amortizes to $O(1)$ carries per successor operation.
	To keep track of the number of successor operations to carry out,
	count down from $M$ to 0 on a separate work tape,
	which by a similar argument takes $O(M)$ time as well.
\end{proof}

\subsection{Uniformity implications}

\linearBDImpliesLogD*
\begin{proof}
	We have to show that $(L_{BD}(C), \pi_4) \in \DTIME(n + f(k))$
	implies that there is some computable $g$ with $(L_D(C), \pi_4) \in \RTIME(\log n + g(k))$.
	We show that there is an $\RTIME(\log n + g(k))$-machine
	that converts its input into the corresponding input for the $\DTIME(n + f(k))$-machine
	and then simulates it to decide $L_{BD}(C)$.

	Let $M$ be a $\DTIME(n + f(k))$ machine deciding $L_{BD}(C)$.
	Assume without loss of generality that $f(k) \geq k$ for all $k$,
	and fix a constant $u$ such that $2^{u(\log n + f(k))}$
	is an upper bound to the admissible numbering of $C$.

	A random-access Turing machine decides $L_D(C)$ as follows.
	It rejects inputs not of the form
	$\langle G, a, p, z, z' \rangle$.
	If the input is of that form,
	it first reads $n = |z|$ and $k = |z'|$.
	It then reads $|G|$, $|a|$, and $|p|$,
	and rejects if any of these lengths exceed $u(\log n + f(k))$.

	It now writes $y = \langle G, a, p, n, k \rangle$ on a scratch tape.
	By merit of the previous checks,
	$y$ has length at most $O(\log n + f(k))$,
	and so this takes at most $O(\log n + f(k))$ time.
	Finally,
	it acts as $M$ on the scratch tape,
	simulating $M$ with input $y$.
	By assumption on $M$,
	this takes $O(|y| + f(k)) = O(\log n + f(k))$ time,
	and $M$ accepts if and only if $\langle G, a, p, n, k \rangle \in L_{DB}(C)$,
	i.e. if and only if $\langle G, a, p, z, z' \rangle \in L_D(C)$.
	The whole procedure runs in $O(\log n + g(k))$ time for some computable $g$.
\end{proof}

\logtimeLinearBDACn*
\begin{proof}
	Let $M$ be a random-access Turing machine of runtime $t(n,k) = a\log n + f(k)$.
	Without loss of generality,
	$M$ does not read from (move its head on) the input tape,
	but only accesses the input via queries:
	it cannot read beyond the first $t(n,k)$ bits by moving its head across the (read-only) input tape,
	so we may just as well begin the computation with queries to copy the first $t(n,k)$ bits of the input to a work tape
	and then treat that as the input tape.
	By \cref{lem:counting},
	the copying (querying while counting up to $t(n,k)$ on the query tape)
	takes $O(\log n + f(k))$ time.
	From hereon, we again assume that $t(n,k)$ is a runtime upper bound.

	Fix a binary encoding of the set $R = \{ 0, 1, \bot \}$ in $\{0,1\}^2$,
	where we interpret 1 and 0 as query responses and $\bot$ as an out-of-bounds response.
	There is a non-random-access multitape Turing machine $M'(r)$ that,
	if given a string $r = r_0r_1\dots r_{t(n,k)-1}$ with each $r_i \in R$,
	acts as $M$ while answering queries based on $r$:
	the $i$th query receives response $r_i$.
	It runs at most $t(n,k)$ steps of $M$ and returns 1 if $M$ returns 1 and 0 otherwise.
	This takes $O(t(n,k))$ time, without random access.

	There is a linear-BD-uniform parameterized circuit family $C$
	that on input $x$ computes $M(x)$ by essentially running $M'(r)$ for all $r \in R^{t(n,k)}$ in parallel,
	and selecting the output where $q$ corresponds with the correct query responses.
	It looks as follows.

	The final gate of $C_{n,k}$ is an \gor/ gate that takes $2^{2t(n,k)} = O(2^{f(k)}n^{2a})$ \gand/ gates,
	each of which takes $t(n,k)+1$ inputs.
	Reading $r$ of length $2t(n,k)$ as a list $r_0 \dots r_{t(n,k)-1}$, each $r_i \in \{0,1\}^2$,
	the $m$th input of the $r$th \gand/ gate for $m < t(n,k)$ is:
	\begin{align*}
		\begin{cases}
			\text{input gate $x_i$} & \text{if the $m$th query made in $M'(r)$ is $i$ and $i < n$ and $r_m = 1$;}\\
			\text{a circuit for $\lnot x_i$} & \text{if the $m$th query made in $M'(r)$ is $i$ and $i < n$ and $r_m = 0$;}\\
			\text{a circuit for constant 0} & \text{if the $m$th query made in $M'(r)$ is $i$ and $i < n$ and $r_m = \bot$;}\\
			\text{a circuit for constant 1} & \text{if the $m$th query made in $M'(r)$ is $i$ and $i \geq n$ and $r_m = \bot$;}\\
			\text{a circuit for constant 0} & \text{if $r_m \notin R$.}
		\end{cases}
	\end{align*}
	The final input for the $r$th \gand/ gate is a circuit for constant 1 if $M'(r) = 1$
	and a circuit for constant 0 otherwise.
	By construction, the output of $C_{n,k}(x)$ is 1 if and only if $M(x) = 1$.

	The circuit family has constant depth and size $O(t(n,k)2^{2t(n,k)}) = O(2^{f(k)}n^{2a+1})$.
	Because $M'$ runs in linear time,
	it is straightforward to check that $C$ is linear-BD-uniform.
\end{proof}

\subsection{Substitution}
\substitutionLemma*
\begin{proof}
	It suffices to define an admissible numbering for $D= A[B/M]$,
	and show that,
	under this numbering,
	$(L_{BD}(D), \pi_4)$ is in $\DTIME(n + f(k))$ for some computable $f$.

	We number the gates of $D_{n,k}$ as follows:
	unmarked internal gates $G$ in $A_{n,k}$ are numbered $\langle 0, G \rangle$ in $D_{n,k}$.
	If $G$ is a marked gate and $G'$ is a gate in the circuit $B_{\fanin(G,n,k),k}$ replacing it,
	we number it $\langle G, G' \rangle$ in $D_{n,k}$.
	The numbering of $D$ is thus based on the numberings of $A$ and $B$,
	and from this it can be shown that our numbering is admissible cf. \cref{def:numbering}.
	The input gates of $D_{n,k}$ must be numbered from $0$ to $n-1$,
	and the output gate must be numbered $n$.
	For simplicity, we instead pretend to number them from $\langle 0, 0 \rangle$ to $\langle 0, n \rangle$.

	We now describe an algorithm that decides $(L_{BD}(D), \pi_4)$
	in $O(n + f(k))$ time on a multitape Turing machine,
	for some computable function $f$. 

	It first rejects inputs $w$ not of the form
	$\langle \langle G, G' \rangle, a, p, n, k \rangle$.
	It further rejects if $G = 0$ and $M(G',n,k) = 1$,
	or if $G \neq 0$ and $M(G',n,k) = 0$:
	if $G = 0$ then $G'$ should not be a marked gate,
	and if $G \neq 0$ then $G$ should be a marked gate.
	By assumption on $M$, this does not take too much time.
	The remaining inputs are handled
	by reducing to a finite number of queries to $L_{BD}(A)$ and $L_{BD}(B)$,
	which can be done in the given time because $A$ and $B$ are linear-BD-uniform.
	They also involve a finite number of computations of $M$ and $\fanin$,
	which by assumption can also be done in the given time.
	The reductions are as follows.

	If $p = \epsilon$ then $w$ is a type query.
	If $G = 0$, then the type of $\langle G, G' \rangle$
	is that of $G'$, so $w$ is accepted if and only if $\langle G', a, p, n, k \rangle \in L_{DB}(A)$.
	If $G \neq 0$, the input is accepted if and only if $\langle G', a, p, n, k \rangle \in L_{DB}(B)$,
	unless $G'$ is an input or output gate,
	in which case they now have the \gor/ type.
	So if $G \neq 0$ and $G' \leq \fanin(G,n,k)$ then $w$ is accepted if and only if $a$ signifies \gor/.

	We now turn to the case $p \neq \epsilon$.
	Then $w$ is a connection query
	that must be accepted if and only if the $p$th input of gate $\langle G, G' \rangle$ is gate $a$.
	The algorithm first rejects if $a$ is not of the form $\langle H, H' \rangle$
	with $H = 0 \iff M(H',n,k) = 0$,
	i.e. if $a$ cannot possibly be a gate in $D_{n,k}$.
	Then there are some case distinctions.

	If $G = H = 0$ then we are dealing with two unmarked gates (or not gates at all),
	and so $w$ is accepted if and only if $\langle G', H', p, n, k \rangle \in L_{BD}(A)$.
	If $G = H \neq 0$,
	then this concerns the internals of a circuit of $B$,
	and $w$ is accepted if and only if $\langle G', H', p, n, k \rangle \in L_{BD}(B)$.

	If $G = 0$ and $H \neq 0$,
	then $w$ is accepted if and only if $H$ is the $p$th input of $G'$ in $A_{n,k}$
	and $\langle H, H' \rangle$ is the output gate of the circuit of $B$ replacing $H$,
	i.e. if and only if $\langle G', H, p, n, k \rangle \in L_{BD}(A)$ and $H' = \fanin(H,n,k)$.

	If $G \neq 0$ and $H = 0$,
	then $w$ is accepted if and only if
	$H'$ is the $p$th input of $G$ in $A_{n,k}$
	and $G'$ is the $p$th input gate of the circuit of $B$ replacing $G$,
	i.e. if and only if $\langle G, H', p, n, k \rangle \in L_{BD}(A)$
	and $G' = p < \fanin(G,n,k)$.

	If $G \neq 0$ and $H \neq 0$,
	then $G'$ should be the $p$th input gate and $H'$ the output gate
	of the circuits of $B$ replacing $G$ and $H$,
	and $H$ should feed into $G$ in $A_{n,k}$.
	So then $w$ is accepted if and only if
	$\langle G, H, p, n, k \rangle \in L_{BD}(A)$
	and $G' = p < \fanin(G,n,k)$
	and $H' = \fanin(H,n,k)$.

	It can be verified that the algorithm is correct
	and runs in time $O(|w| + f(k))$ for some computable $f$,
	thus proving that $D = A[B/M]$ is linear-BD-uniform.
\end{proof}

\subsection{Extended uniformity}
\linearBDImpliesLogtimeE*
\begin{proof}
	The proof is similar to those of \cref{lem:linear-bd-implies-log-d,lem:fo-d-implies-linear-bd},
	in that $C$ will be made up of layers of simgates.
	We make sure that the number of simgates per layer is exactly a power of 2,
	so that the legality of a step between layers depends only on its length.
	Because our list encoding also contains the length of each element
	(cf. \cref{def:list}),
	most steps can be performed by only reading the length instead of the step itself.

	\proofsubparagraph{The circuit family}
	Assume without loss of generality
	that the size and admissible numbering of $A$ is bounded by $f(k)n^c$,
	and its depth by $d(k)$.
	Write $L(n,k)$ for $f(k) + (c+1)\log n$.
	Then $N(n,k) = 2^{L(n,k)} - 1$ is an upper bound to the admissible numbering of $C$,
	and $x \leq N(n,k)$ if and only if $|x| \leq L(n,k)$.

	As in the proof of \cref{lem:linear-bd-uniform-para-ac0=logtime-d-uniform-para-ac0},
	we let $C$ consist of $d(k)$ layers of simgates of exactly $N(n,k)$ simgates each,
	together with a subcircuit to propagate the output of the correct simgate of the final layer.
	Because $A$ is linear-BD-uniform,
	the query blocks are easy to fill in:
	the ``does $q$ take $p$'' blocks can be implemented using \gor/s over $L_{BD}(A)$ queries,
	and the $L_{BD}(A)$ queries can be implemented using constant circuits,
	as $L_{BD}(A)$ queries can be computed in $O(\log n + g(k))$ time for some computable $g$,
	as in the proof of \cref{lem:linear-bd-implies-log-d}.

	We must now be careful in ordering our internal gate inputs in \cref{fig:simgate},
	to make tracing paths tractable.
	We let the first $p \leq N(n,k)$ inputs of gates $\langle m, q, 8 \rangle$ and $\langle m,q,9 \rangle$
	to be $\langle m, q, 12, p \rangle$ and $\langle m, q, 13, p \rangle$ respectively,
	and the remaining $N(n,k) + 1 + i$ for $i < n$ inputs to be gates $\langle m,q,10,i \rangle$ and $\langle m, q, 11, i \rangle$ respectively.

	A gate in $C_{n,k}$ has at most $N(n,k) + n$ inputs,
	meaning that all steps in valid paths in $C_{n,k}$ have length at most $L(n,k) + \log n + 1$.
	Furthermore,
	every valid path contains at most one step of length greater than $L(n,k)$,
	namely the one from gates $\langle m, q, 8 \rangle$ or $\langle m,q,9 \rangle$
	leading to an input gate.
	Finally, because each simgate has constant depth,
	there is some computable $D(k) = O(d(k))$ so that
	all valid paths in $C_{n,k}$ have length at most $D(k)$.

	\proofsubparagraph{Uniformity}
	By a similar proof as in \cref{lem:fo-d-implies-linear-bd},
	it can be shown that $C$ is logtime-D-uniform.
	(It is even linear-BD-uniform.)
	It remains to show that there is an $O(\log n + f(k))$ algorithm for
	recognizing correct paths.

	The algorithm works as follows.
	On input $x$,
	it first rejects if $x$ is not of the form $\langle G, a, p, z, z' \rangle$.
	Let $n = |z|$ and $k = |z'|$.
	Then the input is further rejected if $G$ or $a$ are not gates in $C_{n,k}$.

	Because $p$ is itself part of a list,
	every bit in $p$ is doubled on the tape itself.
	This adds a constant multiplicative overhead
	to the operations that will be performed
	(in particular, addresses need to be doubled),
	that we will ignore in our description below.

	The algorithm will treat $p$ as a valid list encoding,
	and reject if it ever finds that it is not:
	it will reject if any element length is greater than $L(n,k) + \log n + 1$,
	or if after $D(k)$ steps the end of $p$ has not yet been reached.

	It will then, starting from $G$, trace the path $p$
	until one of the rejection conditions above is reached,
	a dead end is reached but $p$ still continues (leading to rejection),
	or the end of $p$ is reached,
	keeping track of where it is in $C_{n,k}$
	to finally compare the end gate with $a$.

	We describe how this tracing can be performed for paths within the $d(k)$ layers of simgates.
	Paths starting in the final output propagation subcircuit can be traced in a similar way.
	We will not go into the internals of the query blocks,
	because they are of constant depth and thus,
	once a path enters one of the query blocks,
	the rest of it is easy to trace in $O(g(k) + \log n)$ time
	for some computable $g$.

	\proofsubparagraph{Path tracing}
	If $\langle m, q \rangle$ is a simgate,
	we call $m$ its \emph{layer} and $q$ its \emph{index}.
	The algorithm traces $p$ as follows.
	It iterates over the steps in $p$, keeping track of the following variables:
	\begin{itemize}
		\item the layer $m$ of the simgate it is in,
		\item an address $a_q$ of where the index of the simgate it is in is stored,
		\item an address $a_q'$ of where the index of the next simgate it may be in is stored,
		\item and the simgate-internal gate indicator $H$.
	\end{itemize}
	That is, if the path traced thus far has led to gate
	$\langle m', q', H' \rangle$ or $\langle m', q', H', v \rangle$,
	then $m = m'$ and $H = H'$, and the list element stored at $a_q$ is $q'$.
	The variables are initialized based on the starting gate $G$,
	and updated as the algorithm moves along $p$.

	Keeping track of $m \in [1, d(k)]$ takes $O(d(k))$ time in total,
	because it changes at most $d(k)$ times,
	and only when moving out of a simgate
	(when $H = 12$ and the next step is $0$).

	Keeping track of $H$ is done as follows.
	If the algorithm is currently on one of the many gates that take at most 3 inputs
	($H \notin \{ 8, 9 \}$),
	then the path is rejected if the length of the next step $p_i$ is greater than 2.
	Otherwise, the value of $p_i$ is read, and $H$ updated appropriately.
	In total, handling these cases also takes $O(d(k))$ time.

	The variables $a_q$ and $a_q'$ are kept track of as follows.
	When $H$ is 8 or 9 and the next step $p_i$
	satisfies $|p_i| \leq L(n,k)$,
	then $a_q'$ is set to $p_i$.
	When moving out of the simgate ($H = 12$),
	$a_q$ gets the value of $a_q'$.
	The values of $a_q,a_q'$ are updated at most $d(k)$ times,
	and the lengths of the addresses are bound by $O(\log(O(f(k) + \log n)))$,
	meaning that in total keeping track takes $O(d(k)(g(k) + \log\log n))$ time for some computable $g$.

	At the end of the path $p$,
	the algorithm constructs the precise gate it traced to
	by reading the value of $q$ stored at address $a_q$,
	and potentially the value $q'$ stored at address $a_q'$ (in case $H \in \{ 12, 13 \}$.)
	It then checks whether that gate is equal to the gate $a$ specified in the input.
	Doing this takes time $O(g(k) + \log n)$ for some computable $g$.

	While we do not address the final tracing in the query blocks (as discussed above),
	we will say a bit more about the case where the path leads out of the simgate to an input gate.
	That is, when $H$ is 8 or 9
	and the length of the next step $p_i$ is greater than $L(n,k)$
	(and smaller than $L(n,k) + \log n + 1$, else the path is rejected).
	The algorithm then reads $p_i$ and computes $p_i - L(n,k)$
	to determine which input gate branch the path leads to.
	This takes $O(g(k) + \log n)$ time for some computable $g$,
	but this happens at most once per path
	(as after this point it leads in constantly many steps to a query block or an input gate).

	\proofsubparagraph{Time}
	It remains to show that the procedure takes $O(\log n + g(k))$ time for some computable $g$.
	Above, we have skipped over the time it takes to read the lengths of the steps in $p$.
	The algorithm reads the lengths of at most $D(k)$ steps in $p$,
	and each length is written in a block of length
	$O(\log(L(n,k) + \log n + 1)) = O(\log f(k) + \log\log n)$ bits
	(as on longer blocks the input is immediately rejected).
	By \cref{lem:counting} (counting on the query tape and querying after every successor operation),
	this takes $O(g(k)\log\log n)$ bits for some computable $g$.
	Apart from reading the step lengths,
	keeping track of the variables $m,a_q,a_q',H$ takes, by the discussion above,
	also $O(g(k)\log\log n)$ bits for some computable $g$.
	Managing the finite number of special steps on a path that we did not elaborate on
	(query blocks, output propagating subcircuit)
	take $O(\log n + g(k))$ time each, for some computable $g$.

	It thus remains to show that $g(k)\log\log n$ is $O(\log n + h(k))$ for some computable $h$.
	That this holds for every combination of $n$ and $k$ follows from a case distinction:
	if $g(k) \leq \log\log n$ holds then $g(k)\log\log n$ is $O((\log\log n)^2)$ which is in $O(\log n)$,
	and if $g(k) > \log\log n$ holds then $g(k)\log\log n$ is $O((g(k))^2)$,
	both of which are $O(\log n + h(k))$ if we set $h(k) = (g(k))^2$.
\end{proof}

\subsection{Beyond \texorpdfstring{$\FO$}{FO}-computable parameters}\label{app:pfo}

Some changes have to be made in the proofs of our results to work with $\pFO$ and general computable parameter functions $\kappa$ instead of $\FO$ and $\FO$-computable $\kappa$.

In the proof of \cref{lem:logtime-d-implies-fo-d},
instead of our \cref{thm:chen},
the original formulation in \cite{chen_lower_2019} has to be used,
Proposition 6, implication $\text{(i)} \implies \text{(ii)}$.

For \cref{lem:fo-d-implies-linear-bd},
we now have to show that $(L_D(C), \pi_4) \in \pFO$
implies $(L_D(C), \pi_4)$ is in linear-BD-uniform $\para\AC^0$.
By the definition of $\pFO$ and \cref{thm:uniform-ac0,thm:regan-vollmer},
$(L_D(C), \pi_4) \in \pFO$ implies that there is a language $Q$ in linear-BD-uniform $\AC^0$
with
\[Q \cap L_D(C) = \{ \langle x, f(\kappa(x)) \rangle : x \in 2^* \}.\]
By \cite[Proposition 6, implication $\text{(ii)} \implies \text{(i)}$]{chen_lower_2019},
$L_D(C)$ is in linear-BD-uniform $\para\AC^0$.

We also prove the equivalence between the two definitions of $\pFO$.
Say that a problem $(Q, \kappa)$ satisfies (F1) if there is an $\FO$-formula $\phi$ and a computable $f$
with
\[Q = \{ x : \langle x, f(\kappa(x)) \rangle \models \phi \}.\]
Say it satisfies (F2) if there is an $\FO$-formula $\psi(N,F)$ and a computable $f$
so that
\[ Q = \{ x : x p(x) \models \phi(|x|, f(\kappa(x))) \}, \]
where $p(x)$ is some (unary) padding function with $|xp(x)| = \max(|x|, f(\kappa(x)))$.

It is straightforward to see that (F2) implies (F1):
taking a formula $\psi$ satisfying (F2),
a formula satisfying (F1) is obtained from it by
adding existential quantifiers to the front of it to capture $N$ as $|x|$ and $F$ as $f(\kappa(x))$,
bounding all quantifiers by replacing $(\forall i)\psi'$ with $(\forall i)((i \leq N \land i \leq F) \implies \psi')$
and $(\exists i)\psi'$ with $(\exists i)(i \leq N \land i \leq F \land \psi')$,
and finally (using quantifiers) replacing all $X(i)$ so that for $i < |x|$ the $i$ are offset with $\delta(|x|) + 2$ bits (cf. \cref{def:list}),
and for $|x| \leq i < f(\kappa(x))$ it reflects $p(x)$.

Going from (F1) to (F2) is slightly more involved,
as in (F1) the input length (and hence the domain) is greater than in (F1)
due to the list encoding overhead.
As $|\langle x, f(\kappa(x)) \rangle|$ is bounded by $B(x,k) = 4\cdot\max(|x|, f(\kappa(x)))$
and we may assume without loss of generality that $B(x,k) > \lceil \log B(x,k) \rceil > 2$,
we can modify an (F2)-formula to work with numbers up to $|\langle x, f(\kappa(x)) \rangle|$ on domain $\max(|x|, f(\kappa(x)))$ by using multiple variables for every number
and then bounding the quantifiers.

We address only how to use two variables to increase the effective domain in general.
Starting from an $\FO$-formula $\phi$,
define $r(\phi)$ recursively by $r(\psi \land \psi') = r(\psi) \land r(\psi')$,
and $r(\lnot \psi) = \lnot r(\psi)$.
For the quantifiers, let $r((\forall i)\psi) = (\forall i_0)(\forall i_1)r(\psi)$.
Finally, let $r(\bit(i,j)) = \bit'(i_0, i_1, j_0, j_1)$, where $\bit'(i_0,i_1,j_0,j_1)$ is defined as
\[
	\bit'(i_0,i_1,j_0,j_1) \iff
	((\bit(i_1, j_1) \land i_1 < |n| - 1) \lor (\bit(i_1, j_0) \land i_1 \geq |n| - 1 \land i_1 \leq 2(|n| - 1))).
\]
That is, the first $|n| - 1$ bits of our composite number are stored in $j_1$ and the second $|n| - 1$ bits are stored in $j_0$.
Now $r(\phi)$ can on word models of length $n$ work with numbers of $2^{2|n| - 2}$.
(There are additional technicalities with regards to $<$ and $X$, but they are resolved similarly.)

\end{document}